\documentclass[10pt,journal,compsoc,romanappendices]{IEEEtran}
\ifCLASSOPTIONcompsoc
  \usepackage[nocompress]{cite}
\else
  \usepackage{cite}
\fi
\usepackage{tablefootnote}

\usepackage{lineno,hyperref}
\modulolinenumbers[5]

\usepackage{graphicx}
\usepackage{url}
\usepackage{bbm}
\usepackage{amsfonts}
\usepackage{amssymb}
\usepackage{kantlipsum}
\usepackage{enumerate}
\usepackage{amsthm}
\usepackage[T1]{fontenc}
\usepackage{multirow}
\usepackage{amsmath}
\usepackage{bbm}
\usepackage{dsfont}
\usepackage{xcolor}
\usepackage{booktabs}
\usepackage{mathtools}
\usepackage[utf8]{inputenc}
\usepackage{enumerate}
\usepackage{algorithm}
\usepackage{algpseudocode}
\newcommand{\remove}[1]{}

\newtheorem{lemma}{Lemma}[section]
\newtheorem{remark}{Remark}[section]

\newtheorem{proposition}{Proposition}

\algdef{SE}[DOWHILE]{Do}{doWhile}{\algorithmicdo}[1]{\algorithmicwhile\ #1}

\def\E{{\mathbb{E}}}
\def\P{{\mathbb{P}}}

\algnewcommand{\algorithmicgoto}{\textbf{go to}}%
\algnewcommand{\Goto}[1]{\algorithmicgoto~\ref{#1}}
\begin{document}

	\title{ Slotted ALOHA and CSMA Protocols for \\FMCW Radar Networks}
	\author	
	{		
		\IEEEauthorblockN		
		{		
			Haritha~K\IEEEauthorrefmark{1},
			Vineeth Bala Sukumaran\IEEEauthorrefmark{2},
			Chandramani~Singh\IEEEauthorrefmark{1}
		}

		\IEEEauthorblockA		
		{
			\IEEEauthorrefmark{1} Indian Institute of Science, Bangalore, India \\ \{haritha, chandra\}@iisc.ac.in}
		\IEEEauthorrefmark{2} Indian Institute of Space Science and Technology, Trivandrum, India  \\ \{vineethbs\}@iist.ac.in}

\maketitle

\begin{abstract}
 We study medium access in FMCW radar networks. We assume that all the radars use the same parameters, e.g., chirp duration, chirp slope, cutoff frequency, number of chirps per packet, etc, and propose and analyze slotted ALOHA and CSMA protocols to mitigate narrowband interference. We define a notion of throughput to quantify the performance of the proposed protocols. In the case of ALOHA, we analyze interference probability and throughput as functions of the system parameters. We observe that interference probability and throughput may behave differently than in wireless communication networks. For instance, if the number of chirps per packet is larger than one, the interference probabilities may be smaller for higher transmission rates. We define a medium sensing procedure, referred to as clear channel assessment (CCA), as a part of the proposed CSMA, and also define CCA success and failure events. In CSMA, the radars transmit only after a successful CCA. We study, CCA success probability, interference probability, and throughput as functions of the system parameters. We observe that, unlike wireless communication networks, using the highest possible attempt rates may maximize throughput in a few network scenarios. We perform an extensive simulation to verify our analytical results and to compare slotted ALOHA and CSMA. We observe that CSMA outperforms ALOHA in all realistic scenarios. 
\end{abstract}

\section{Introduction}
Self-driving cars, also referred to as driverless or autonomous cars are poised to become a reality in the next five to ten years. To ensure the safety and efficiency of transportation systems, self-driving cars must have a $ 360  $-degrees view of their surrounding in a way that is ecologically and economically sustainable. These cars need various sensors such as optical cameras, millimeter-wave radars, and laser imaging detection and ranging to map the environment including cars, obstacles, and pedestrians surrounding them. 
\par Millimeter-wave Frequency Modulated Continuous Wave (FMCW) radars~\cite{TI_1} are long-range radars that can reliably detect distance, velocity, angle of other vehicles and objects in the vicinity self-driving cars. These radars operate in the $76$-$77$ GHz band~\cite{TI_2} with a bandwidth of $ 1 $ GHz. These can be made with compact antennas while retaining sufficiently high gain and narrow beamwidth to offer high range and angular resolution; they can offer localization sensitivity up to $3 $ cm. FMCW radar electronics are robust, and radar frequencies are generally immune to adverse weather conditions (snow/fog/rain or optical effects)~\cite{TI_2}. 

\par For reliable operation of FMCW radars, it is imperative that received radar signals be interference-free. Interference in $ 76$-$77 $ GHz band will be a crucial challenge in near future in view of the likely increase in: (a) the number of self-driving cars from multiple manufacturers, (b) the number of radars per car, and (c) the operating duty cycle per radar. There are two kinds of inter-radar interference \cite{Inter_Mit_Aydogdu}, \cite{Interference_Rao}. $ (1) $ Wideband interference: This is applicable to heterogeneous radar networks in which different radars use different chirp slopes.
This manifests itself as an increase in noise floor, leading to \textit{signal-to-noise-ratio} (SNR) degradation, and therefore misdetection of targets. $(2)$ Narrowband interference: this happens if a radar has another (interfering) radar using the same chirp slope in its vicinity and the beat signal frequency due to interfering radar is less than the cut-off frequency of the Low Pass Filter (LPF). This causes ghost detections (or, false alarms)~\cite{TI_1}. Both types of interference deteriorate radar performance, but mitigation of narrowband interference requires signal processing as well as medium access techniques \cite{Sun}. In this work, we focus on the mitigation of narrowband interference using medium access techniques. In this work, we focus on the mitigation of narrowband interference.
\par Medium access techniques have been widely studied in the context of wired and wireless communication networks. In communication networks, if a node receives signals from two simultaneously transmitting nodes, it cannot decode any of the signals; this phenomenon is termed as collision. A number of centralized~(e.g., TDMA) and distributed~(e.g., ALOHA, CSMA, CSMA-CA) multiple access protocols have been proposed to minimize interference in communication networks. In the context of FMCW radars, medium access control can reduce false alarms but may also lead to missed targets. Design and analysis of medium access control protocols is pivotal for large scale deployment of radar networks. 
\subsection{Related work}

 Jin et al.~\cite{vtc_radar} considered a network of FMCW radars with identical chirp parameters and obtained statistical characterizations of interference signal properties under an asynchronous pure random access scheme. A received chirp is classified as a false alarm or a real target based on its received power.
 Further, they considered continuous time setup and the radars chose transmission start time randomly and after that, they transmit chirps continuously. In our work, we consider slotted time structure and there is a random time duration between the successive packets.
 
\par In \cite{csma_radar} and \cite{Inter_Replica_Ammen} the authors use carrier sensing to mitigate the interference. Ishikawa et al.~\cite{csma_radar} considered collocated radars with identical parameters. The authors proposed a carrier sense multiple access (CSMA) algorithm to avoid the interference of FMCW radar signals by shifting the transmission timing after sensing the medium. But, they do not analyze the performance of the algorithm when the radars are non-collocated. We propose a variant of CSMA and discuss its performance in both collocated and non-collocated radar networks.
\par Ammen et al.~\cite{Inter_Replica_Ammen} propose a technique that generates an interference replica of the ghost signal during carrier sensing. This is subtracted from the received signal in the frequency spectrum to mitigate narrowband interference.
\par To mitigate narrowband interference different techniques, e.g., TDM, FDM, phase coding, frequency hopping, and chirp sequence are used to orthogonalize the chirp transmissions in \cite{radar_chat}, \cite{Jin_Journal}, \cite{FH_Luo}, \cite{Interference_Rao}, \cite{radar_mac}, \cite{Centra_Mazher}, and \cite{Son}.
\par Centralized medium access mechanisms for radars were proposed in \cite{radar_mac} and \cite{Centra_Mazher}. In \cite{radar_mac}, a control center was setup to receive speed and location information from all the nearby radars. Using these the controller computes the waveform parameters that provide orthogonal access to the medium and dispatches them to each one to avoid interference. In \cite{Centra_Mazher}, the authors proposed a resource allocation strategy that relies on the existing communication infrastructure to jointly
assign slope directions and carrier frequency offset to all the radar units operating in its vicinity.
\par Aydogdu et al.~\cite{radar_chat} proposed a medium access control protocol RadChat. It is a coordinated framework among the nearby vehicles and uses time division and frequency division techniques to mitigate interference caused by neighboring radars. The schemes discussed in \cite{radar_chat} and \cite{radar_mac} use time synchronous orthogonal techniques to avoid interference. The challenge in these time synchronous schemes is propagation delay~($\approx$us) can deteriorate the efficacy of orthogonality.
\par Son et al.~\cite{Son} proposed an algorithm where they assign a chirp sequence for FMCW radar. They design a chirp sequence set such that the slope of each vehicle's chirp sequence will not overlap within the set. By assigning one of the chirp sequences to each vehicle, interference can be mitigated.
\par Jin et al.~\cite{Jin_Journal} proposed asynchronous non-cooperative protocols to mitigate interference in radars. They quantify the network performance using cross layer performance metrics such as multiple access capacity and target misdetection probability. They explore different combination of frequency hopping, phase coding, and random frequency division multiple access techniques to avoid interference.
\par Luo et al.~\cite{FH_Luo} proposed a frequency-hopping random chirp technique that reconfigures the chirp sweep frequency and time in every cycle. This results in a noise-like frequency response after the received signal is demodulated at the receiver. Rao et al.~\cite{Interference_Rao} used clock drift that changes the slope of interfering radar with respect to the tagged radar. Further, they employ binary phase encoding to mitigate interference.
\par Al-Hourani et al.~\cite{stoch_geometry} considered a Poisson point process and a Bernoulli lattice process as models for the spatial distribution of radars. Using stochastic geometry they compute interference statistics and obtain analytical expressions for the probability of successful range estimation.

\subsection{Our Contribution}
We consider slotted ALOHA protocol, wherein the radar transmits for a fixed amount of time and enters into backoff state. We propose CSMA, in which each radar senses the medium before transmitting the packet. 
\begin{itemize}
	\item We derive the probability of interference in FMCW radar networks assuming that all the radars are synchronized and collocated in ALOHA. 
	
	\item We also derive the probability of interference in the scenario where the radars are not collocated and their clocks are not synchronized, i.e., each pair of radars is associated with a non-zero propagation time in ALOHA.
	 
	\item In both the cases, as expected, the interference probability increases with medium access probability if the number of chirps per packet is one. However, if packets consist of more than one chirp, the interference probability first increases and then decreases with medium access probability.
	
	\item We define a notion of throughput and derive throughput for both of the above scenarios. If each packet consists of exactly one chirp, the throughput first increases and then decreases with medium access probability. In this case, we derive the expression for the optimal medium access probability. When the packets consist of more than one chirp the throughput first increases and reaches the maximum, then decreases and increases with medium access probability.	
	
	\item We propose a CSMA protocol that senses the medium before transmitting a packet. we define a medium sensing procedure called clear channel assessment (CCA) and CCA success and failure events. Further, we analyze the CSMA protocol using fixed point analysis for a restricted class of radar networks.
	 
	\item We perform extensive simulation to verify our analytical results and to compare slotted ALOHA and CSMA.
\end{itemize}
\paragraph*{Organization of the paper}
We explain the operation of a FMCW radar in Section~\ref{subsec:fmcw radars} and introduce the system model in Section~\ref{subsec:Radar Network}. In 
Section~\ref{sec:Performance_analysis} we analyze slotted ALOHA and derive the probability of interference and throughput. In Section~\ref{sec:CSMA Algo} we propose CSMA protocol. In Sections~\ref{subsec:collocated-Async} and \ref{subsec:csma-non-coll} we analyze CSMA in collocated and non-collocated radar networks respectively. In Section~\ref{sec:numerical_res} we verify our analytical results obtained and compare the performance of slotted ALOHA and CSMA protocols.
   \section{System Model}
\subsection{FMCW Radars}\label{subsec:fmcw radars}
In this section we explain the operation of a FMCW radar and narrowband interference using radar transmit/receive block Figure~\ref{fig:block-diagram} \cite{vtc_radar}. The FMCW radar is used to measure the range, velocity, and angle of arrival of a target in front of it. The radar generates a sinusoid signal called chirp whose frequency increases linearly with time. The chirp is characterized by following parameters, starting frequency~$(f_{\min})$, ending frequency~$(f_{\max})$, bandwidth~$(B)$, chirp duration $(T_c)$, and slope~$(h=\frac{B}{T_c})$ as shown in Figure~\ref{fig:chirp}.
\par To explain the operation of FMCW radar we consider a single target in front of the radar at distance $d$. The MAC scheduler schedules the chirps transmissions and the synthesizer generates the chirps at these scheduled times. The transmitting antenna transmits them. The reflected chirp from the target is received by receiving antenna with a delay $\tau = \frac{2d}{c}$. The reflected and transmitted chirps are given as inputs to a mixer. The mixer is a device that takes two sinusoids as inputs and generates another sinusoid whose instantaneous frequency is equal to the absolute difference of instantaneous frequencies of the input signals. Note that the reflected chirp is delayed by $h\tau$ from the transmitter chirp as shown Figure~\ref{fig:chirp}. The target at distance $d$ generates an IF signal with frequency $\frac{2dh}{c}$ at the output of the mixer. This is further passed through a LPF with cut-off frequency $f_H$. The cut-off frequency is a function of the sampling frequency of the ADC. The IF signal is further digitized using ADC and converted into frequency domain using DFT. The detector detects a peak at frequency $\frac{2dh}{c}$. Note that the maximum frequency that can be observed is limited by the cut-off frequency of the LPF. This leads to,
\begin{align}\label{eq:d_max,Delta}
h\tau &< f_H \nonumber \\
\tau < \frac{f_H}{h}&\text{   and   }d<\frac{cf_H}{2h}.
\end{align}

\paragraph*{Target range $(d_{\max})$} The maximum distance between the radar and the target such that the corresponding IF signal is passed through the LPF. This is given by $\frac{cf_H}{2h}$.
\paragraph*{Slot $(\Delta)$} The maximum round trip delay between the radar and the target such that the corresponding IF signal is passed through LPF and is given by $f_H/h$.

\begin{figure}[H]
	\centering
	\includegraphics[width=8cm, height=2.75cm]{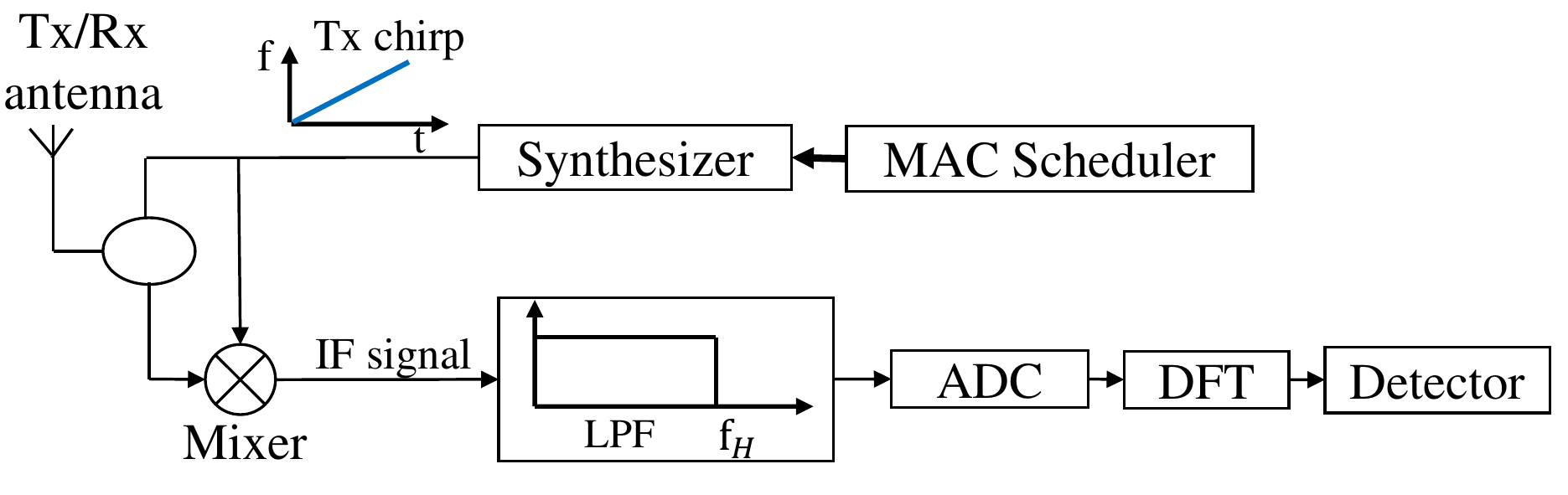}
	\caption{TX/RX block diagram for a FMCW radar}
	\label{fig:block-diagram}
\end{figure}
\begin{figure}[H]
	\centering
	\includegraphics[width=4cm, height=3cm]{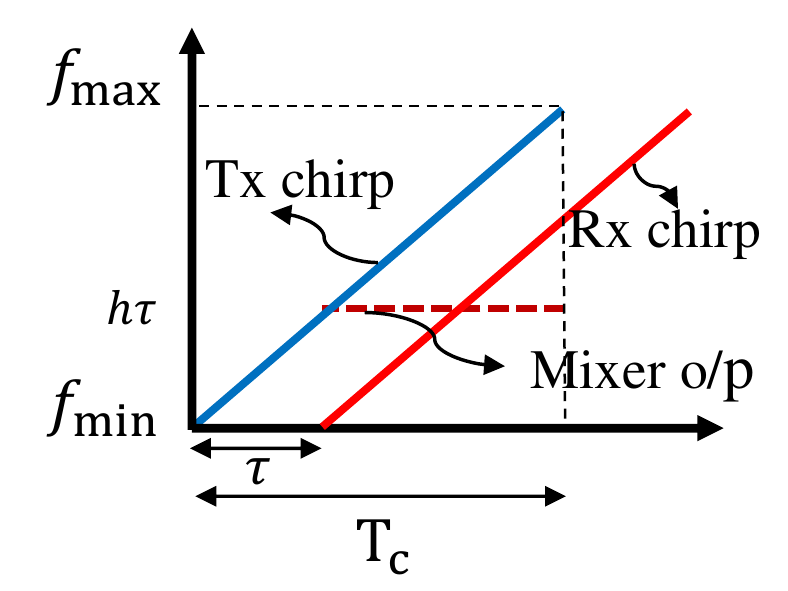}
	\caption{Illustration of the transmitted chirp, received chirp, and the IF signal obtained by mixing
	}
	\label{fig:chirp}
\end{figure}

\paragraph*{Multiple targets in front of radar} Now consider the scenario where there are two targets in front of the radar with in the target range and we want to compute their distance from the radar. Let $d_1, d_2$ be the distances from the radar. Ideally we should be able to observe two peaks in the frequency spectrum at $\frac{2hd_1}{c}, \frac{2hd_2}{c}$ corresponding to the targets at distance $d_1, d_2$ respectively. The maximum amount of time available to observe the reflected chirps is $T_c$ (see Figure~\ref{fig:chirp}). This observation window $T_c$ can distinguish frequency components that are separated by more than $\frac{1}{T_c}$Hz \cite{TI_1}. So, to observe two peaks corresponding to the distances $d_1$ and $ d_2 $ in the frequency spectrum,
\begin{align}\label{eq:range_resolution}
\frac{2hd_1}{c}-\frac{2hd_2}{c} &> \frac{1}{T_c} \nonumber \\
d_1-d_2 > \frac{c}{2B}
\end{align}
where $B$ is the bandwidth of the chirp. 
\paragraph*{Range Resolution $(\Delta d)$}
Range resolution is the minimum distance between the two targets such that the radar can distinguish them and it is given by $\Delta d = \frac{c}{2B}$.
\begin{remark}
If the distance between two targets is less than the range resolution then we observe only a single peak in the DFT spectrum and radar detects them as a single target.
\end{remark}

\paragraph*{False alarm or interference}\label{para:Interference}
Now consider two FMCW radars with identical parameters and separated by a distance $d$. Let us label these as first radar and second radar respectively. Each radar receives chirps transmitted by the other radar as shown in Figure~\ref{fig:Interference_1}. 
\par Let $t_1, t_2$ be the start times of chirps at the first radar and the second radar respectively. Then at time $t_3=t_2+\frac{d}{c}$ the first radar receives the chirp transmitted by the second radar. If $\tau = |t_1-t_3|< \Delta $ (see Figure~\ref{fig:Interference_1}a) a signal at frequency $f_{IF} = h\tau$ is observed at the output of LPF. Note that the first radar assumes the chirp received from the second radar as reflected version of its own chirp and computes the distance. The first radar assumes that there is a target at distance $d_g=\frac{cf_{IF}}{2h}$. We call this a ghost target because there is no real target at this distance. The phenomenon of a radar detecting a ghost target is called a \textit{false alarm}. We also call this \textit{interference}. Note that the first radar does not observe a false alarm for every chirp received from the second radar. The following time condition has to be satisfied to observe a false alarm at the first radar. 
\begin{align}\label{eq:time_condition}
t_1-\Delta < t_2+\frac{d}{c} < t_1+\Delta.
\end{align}
We define the \textit{Interference range} $(d_{I,\max})$ as the maximum distance between two radars such that the interference is observed. It is a function of chirp power transmitted by the second radar. Any radar outside the interference range contributes no interference. 
\begin{figure}	
	\centering
	\includegraphics[width=4cm, height=4cm]{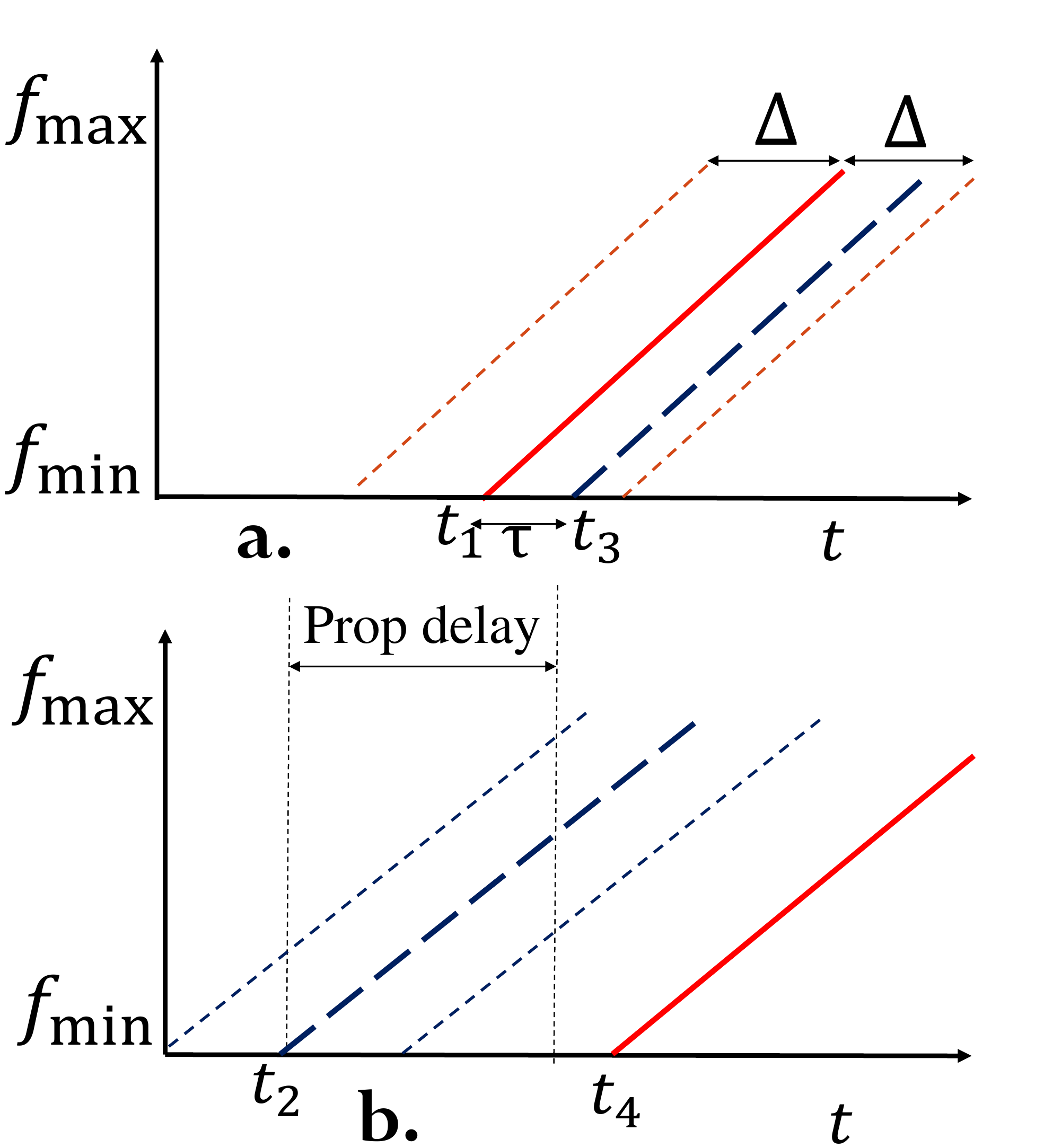}
	\caption{Illustration of a "ghost target" being detected at radar $ 1 $. In (a) and (b) we illustrate the transmitted and received chirps of radars $ 1 $ and $ 2 $. If a radar (radar $ 1 $) receives a transmission within $(t_1-\Delta, t_1+\Delta)$ then we have interference and a ghost target is detected.}	
	\label{fig:Interference_1}
\end{figure}

  \begin{table}
  \begin{center}
			
			\begin{tabular}{||l | l||}			
				\hline \rule{0pt}{9pt}
				Symbol & Description  \\ 
				\hline\hline\rule{0pt}{9pt}
				$f_{\min}$ & Chirp starting frequency\\
				\hline \rule{0pt}{9pt}
				$f_{\max}$ & Chirp ending frequency\\
				\hline \rule{0pt}{9pt}
				$B$ & Chirp bandwidth\\
				\hline \rule{0pt}{9pt}
				$h$ & Chirp slope\\
				\hline \rule{0pt}{9pt}
				$T_c$ & Chirp duration in seconds\\
				\hline \rule{0pt}{9pt}
				$ M $ & Number of radars  \\
				\hline \rule{0pt}{9pt}
				$ K $ & Chirp duration in number of slots  \\
				\hline \rule{0pt}{9pt}
				$ L $ & Number of chirps per packet \\
				\hline \rule{0pt}{9pt}
				$ W $ & Number of backoff slots  \\
				\hline				
			\end{tabular}
		
		 \caption{ Notations and symbols used}
		\label{tab:Table 1}  
	\end{center}
		\end{table}			

\subsection{FMCW Radar Networks} \label{subsec:Radar Network}
We consider a network of $M$ radars with identical chirp parameters. The radars operate using a slotted time structure with slot size~$\Delta$. We assume that the chirp duration is an integer multiple of slot size, i.e., $T_c=K \Delta,$ where $K$ is a positive integer. Further, we define the clock offset $\delta_{k,l}$ of radar $ l $ with respect to radar $ k $ as the time duration between a radar $k$'s clock tick and a radar~$l$'s clock tick that comes immediately after radar~$k$'s clock. Note that $ \delta_{k,l} \in [0, \Delta)$. We call a network \textit{synchronous} if $\delta_{k,l} =0, \forall k,l \in \{1,2,\dots M\}$ else we call it \textit{asynchronous}. As described in Section~\ref{para:Interference}, radars' transmissions interfere with each other. The radars employ a MAC protocol to schedule their transmissions for the possible reduction of interference. In general, in a MAC protocol, each radar interleaves transmissions with silence periods, termed as backoffs. The backoff durations are random whereas a fixed number of, say $L$, chirps are contiguously transmitted before going for the next backoff. We refer to the contiguous sequence of $ L $ chirps as a packet. 
\par We introduce two MAC protocols, namely slotted ALOHA and CSMA, in Sections~\ref{sec:Performance_analysis} and \ref{sec:CSMA Algo}, respectively.  
Here, we define a metric, called throughput, to quantify the performance of a MAC protocol. It is the fraction of time slots spent by all the radars in transmitting packets successfully without interference.
\paragraph*{\textbf{Throughput $(\Theta)$}}\label{para:Theta}
Let $S(T)$ be the total number of packets transmitted successfully by all the radars in a network in $T$ slots. Then the throughput of the network is given by,
\begin{align}
\Theta = \lim_{T\rightarrow \infty} KL \frac{S(T)}{T},
\end{align}
where $K$ is the chirp length in slots and $L$ is the number of chirps per packet.
\paragraph*{Effect of backoff on target detection}
Let us define a backoff and the subsequent packet transmission as a transmission cycle~(see Figure~\ref{fig:Tx_Cycle}). The length of a transmission cycle is typically much shorter than $d_{\max}/v_{\max}$ where $v_{\max}$ is the maximum relative speed of the target; it is the time derivative of the distance between the radar and the target. For instance, consider the following example with typical radar parameters and vehicle speeds \cite{csma_radar}. The slot~$\Delta=\frac{4}{3}\text{x}10^{-6}$~s and $d_{\max}= 200$m. Let $W$ be the number of backoff slots. The change in distance between target and radar in this duration is $\Delta W v_{\max} $ meters. We consider typical values, $v_{\max}=50$m/s ($180$kmph), $W=100$. The distance traveled by the target in $W$ slots is $0.67$cm. We have used typical values for radar parameters and vehicle speeds in this calculation, the actual values may differ but would be of the same order. Consider a target just entered into the radar range when the radar is in backoff state. Let $d$ be the distance between them. The change in $d$ before and after backoff is at most in the order of centimeters and very small compared to $d_{\max}\approx 200$m. So, it takes several transmission cycles for a target to come alarmingly close to the radar after it enters the transmission range of the radar. Therefore, the radars are unlikely to miss the targets because of backoffs.
\par We tabulate the symbols and notations used in TABLE~\ref{tab:Table 1}. Further, we use $\E[.]$ for expectation and Geo~$(p)$ for geometric distribution with parameter $p$.
\section{Slotted ALOHA} \label{sec:Performance_analysis}
\paragraph*{Slotted ALOHA}
We consider a slotted ALOHA protocol wherein each radar transmits $L$ chirps contiguously and enters into backoff for $W$ slots (see Figure~\ref{fig:Tx_Cycle}) and repeats the same. 
Successive backoffs constitute a sequence of independent and identically distributed~(i.i.d) random variables, each having cumulative distribution function $F$. Let $W$ denote a generic random variable with distribution $F$. We assume that $\E[W]$ is finite. 
\begin{figure}	
	\centering
	\includegraphics[width=6.5cm, height=1.8cm]{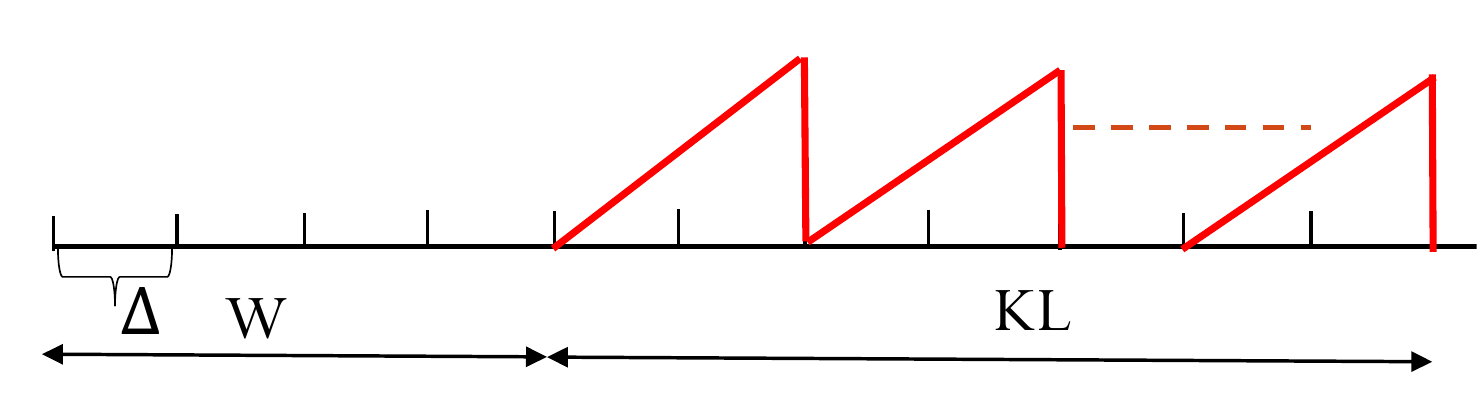}
	\caption{Transmission cycle}
	\label{fig:Tx_Cycle}
\end{figure}
In ALOHA each radar transmits packets irrespective of the other radars' transmissions in the network. So, each radar's packets can be interfered by other radars' packets. 
\par In this section, we analytically characterize the throughput of slotted ALOHA using the probability of interference. We assume all radars in the network have identical chirp parameters. Further, the clock offset of radar~$l$ with respect to radar~$k$ is given by $\delta_{k,l}$ and the distance between them is given by $d_{k,l}$. We compute the probability of interference in the following network scenarios.
\begin{enumerate}
	\item We assume that all the $M$ radars are collocated, i.e., we ignore the propagation delays and that their clocks are synchronized i.e., $\delta_{k,l}=0, \text{ for all }k,l$. We use the insights from this scenario to analyze the second scenario, described below.
	
	\item We assume that all the $M$ radars are non-collocated and each radar $k, k \in \{1,\dots,M\}$ can potentially interfere with all the other radars $l, l\in\{1,\dots , M\}\setminus k$ in the network. Further, we assume clocks are not synchronized.
\end{enumerate}
\par In order to derive the probability of interference, we define the state of a radar in a slot as follows. We say that a radar is in state $0$ if it is in backoff. Further, if a radar is transmitting $i$-th slot in the packet then we say that it is in state $i, i \in \{1, \dots KL\} $. Recall that in ALOHA protocol each radar transmits packets irrespective of other radars' packet transmissions. So, the state evolution of each radar is independent of other radars. Let us define $\pi = (\pi_0, \pi_1, \pi_2, \dots \pi_{KL})$ where $\pi_i, i \in \{0,1, \dots KL\}$ is long term fraction of slots the radar spends in state $i$. We evaluate $\pi$ in the following and use it in the Propositions~\ref{prop:Prop_1}, \ref{prop:Prop_2} and \ref{prop:theta} to compute the probability of interference and throughput.

\par The radar's state evolves according to a renewal process with the slots in which the radar state is $1$ being the renewal epochs. The duration of a transmission cycle $C_i, i \in \{1,2,\dots\}$ is random and independent across $i$. Further, $\mathbb{E}[C_i] = KL+\mathbb{E}[W]$. With this we define renewal epochs at a radar 
\[Z_j =\sum_{i=1}^j C_i, \text{ and }Z_0 = 0\]
and renewal process \[\phi(T) = {\sup\{j \geq 0:Z_j\leq T\}},\] where $\phi(T)$ denotes the number of radar transmission cycles occurred until time slot $T$. Let $R_i$ be the reward associated with $i$-th transmission cycle. The total reward earned in $\phi(T)$ cycles is given by
\[R(T) = \sum_{i=1}^{\phi(T)} R_i(T).\]
We define this reward as the amount of time the radar has spent in transmitting chirps i.e., $R_i = KL, i \in \{1,2, \dots \phi(T)\}$.
From the Elementary Renewal theorem and Renewal Reward theorem~\cite{AK}
\begin{align}\label{eq:RRT}
\lim_{T\rightarrow \infty}&\frac{\phi(T)}{T} = \frac{1}{KL+\mathbb{E}[W]}\text{  almost surely}, \nonumber\\
\lim_{T\rightarrow \infty}&\frac{R(T)}{T} = \frac{KL}{KL+\mathbb{E}[W]}\text{  almost surely }.
\end{align}
We defined
\[\pi_i = \lim_{T\rightarrow \infty} \frac{\text{The amount time spent in state } i \text{ until slot } T}{T}
,\] using \eqref{eq:RRT} in above we write,
\begin{align}\label{eq:pi}
\pi_0 &=\frac{\mathbb{E}[W]}{KL+\mathbb{E}[W]},\nonumber \\
\pi_1 &= \pi_2 =\dots = \pi_{KL} = \frac{1}{KL+\mathbb{E}[W]}.
\end{align}

Now we consider the scenarios described at the beginning of this section to compute interference probability. First, we compute the probability of interference in two radar network and using this result we compute in $M$ radar network. In two radar network we label the radar at which we are observing as tagged radar and the other one as interfering radar. Let $d$ and $\delta$ be the distance between the two radars and clock offset of interfering radar with respect to the tagged radar respectively. We first consider synchronous and collocated radar network. We consider asynchronous and non-collocated radar network subsequently.
\subsection{The case $\delta = 0, d = 0$}\label{subsec:1_prop}
In this case the time condition~\eqref{eq:time_condition} for interference at the tagged radar in the slotted system reduces to 
\[t_1-1 < t_2 <t_1+1, \]
where $ t_1,t_2 $ indicate chirp start time slot indices at the tagged and interfering radars respectively. As $t_1, t_2 $ are integers the above condition gives $t_1=t_2$, i.e., if two chirps at the tagged and interfering radars start at the same slot then tagged radar suffers interference. Recall that~(see Figure~\ref{fig:Tx_Cycle}) each radar transmits $L$-chirps contiguously in one transmission cycle. So, the interference can be caused by any one of the $L$-chirps from interfering radar to any one of the $L$-chirps at the tagged radar as shown in~Figure~\ref{fig:1_Prop}. Further, we say that a packet is successfully transmitted if none of its chirps are interfered by the packets from other radars. 
\par Without loss of generality let us assume a packet transmission started at the tagged radar at slot $t_1$. Let us define following events at the interfering radar $\{E_{-(L-1)}, \dots E_0, E_1 \dots E_{L-1} \}$ where event $E_i, i \in \{-(L-1), \dots L-1\}$ is starting packet transmission at slot $t_1+iL$. Further, observe that these are the only events that can cause interference at the tagged radar. 
\par This is explained for the case $L=2$ in~Figure~\ref{fig:1_Prop}. The interfering events possible in are $ \{E_{-1},E_0,E_1\}$. We have shown events $ E_{-1},E_1 $ starting packet transmission at slots $t_1-2, t_1+2$ respectively. In the case of $E_{-1}$ first and second chirps from the tagged and interfering radars started at the same slot. Similarly, in $E_1$, second and first chirps from tagged and interfering radars started at the same slot. 
\begin{figure}	
	\centering
	\includegraphics[width=8cm, height=4.5cm]{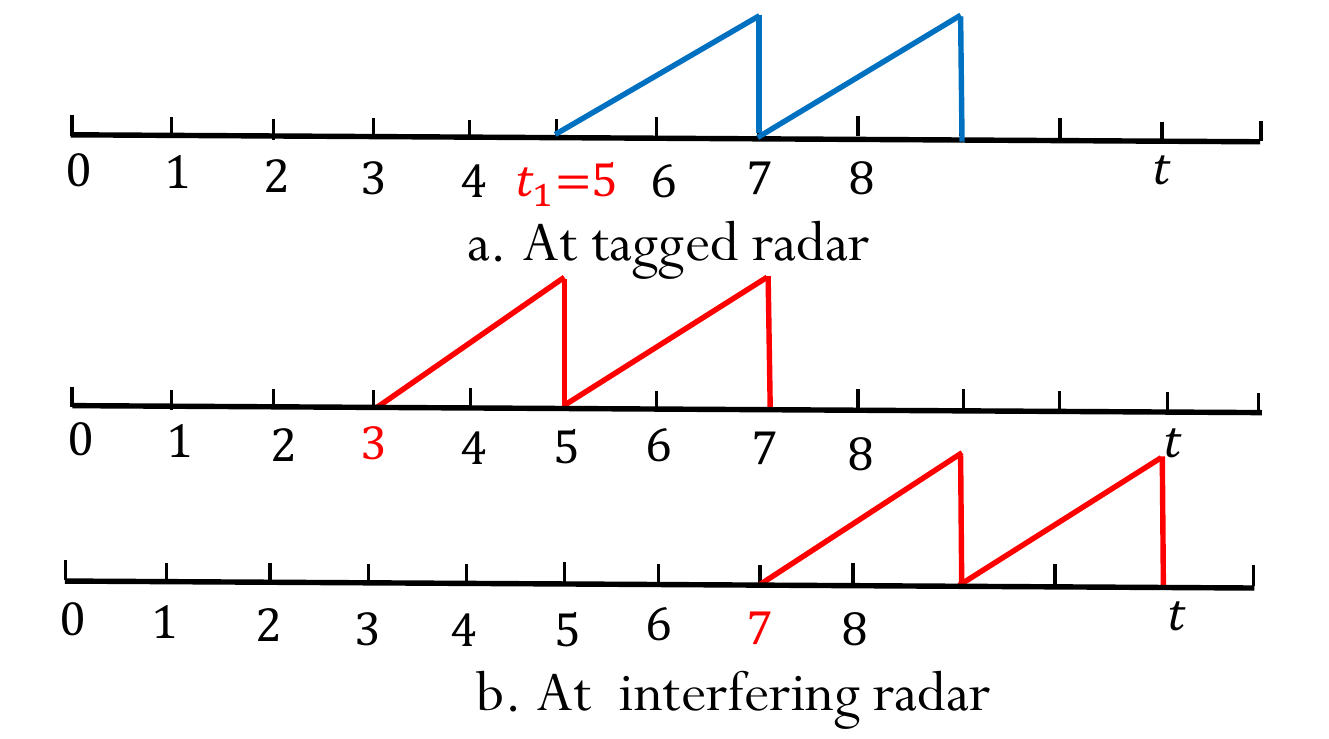}
	\caption{Illustrating the interference events $E_{-1}, E_1$ for $K=2,L=2$, a. The tagged radar transmitting a packet, b. Interfering events at the interfering radar }
	\label{fig:1_Prop}
\end{figure}
\par Given the tagged radar started transmission cycle at slot $t_1$ the only events that can cause interference are $\{E_{-(L-1)}, \dots E_0, E_1 \dots E_{L-1} \}$. From \eqref{eq:pi} we have the probability of starting packet transmission in any slot is $\P(E_i) = \frac{1}{KL+\mathbb{E}[W]}, i \in \{-(L-1)\dots L-1\}$. Using this we define the probability of interference at the tagged radar as
\begin{align*}
p_{I2}=\mathbb{P}\Big(\cup_{i=-(L-1)}^{L-1}E_i\Big).
\end{align*}

\par In Proposition~\ref{prop:Prop_1}, we derive the probability of interference at the tagged radar in two radar network and extend that to the case of a $ M $-radar network.
\begin{proposition}\label{prop:Prop_1}	
\begin{enumerate}
	\item The interference probability at the tagged radar in two radar network is given by,
		\begin{align*}
		p_{I2} = &\frac{2L-1}{KL+\mathbb{E}[W]}\\
		&-\frac{1}{KL+\mathbb{E}[W]}\sum_{i=1}^{L-1}(L-i)\mathbb{P}(W=(i-1)K).
		\end{align*}
	\item The interference probability at the tagged radar in $M$ radar network is given by,
		\begin{align*}
		p_{I} = 1-(1-p_{I2})^{M-1}.
		\end{align*}
\end{enumerate}
	
\end{proposition}
\begin{IEEEproof}	
We observe that $\{E_{-(L-1)}\dots E_0\}$ are mutually exclusive and $\{E_0, \dots E_{(L-1)} \}$ are mutually exclusive. Further, $E_i, i \in \{1,\dots L-1\} $ and $E_{-(L-j)},  j \in \{1, \dots i\}$ are not mutually exclusive.
Now consider,
\begin{align}\label{eq:prob_inter_1}
p_{I2}=\mathbb{P}&\big(\cup_{i=-(L-1)}^{L-1}E_i\big) \nonumber \\ 
	= &\sum_{i=-(L-1)}^{L-1}\mathbb{P}\big(E_i\setminus  (\cup_{j=-(L-1)}^{i-1}E_j)\big)\nonumber
	\\ 
	\stackrel{(a)}{=}&\sum_{i=-(L-1)}^0\mathbb{P}(E_i)+\sum_{i=1}^{L-1}\mathbb{P}\big(E_i\setminus  (\cup_{j=-(L-1)}^{i-1}E_j)\big) \nonumber\\
	=&\sum_{i=-(L-1)}^{(L-1)}\mathbb{P}(E_i)-\mathbb{P}\Big(E_i\cap(\cup_{j=-(L-1)}^{i-1}E_j)\Big) \nonumber\\
	\stackrel{(b)}{=}&\sum_{i=-(L-1)}^{L-1}\mathbb{P}(E_i)-\sum_{i=1}^{L-1}\sum_{j=1}^i\mathbb{P}(E_i\cap E_{-(L-j)}).
	\end{align}
	The $(a)$ and $(b)$ follows from mutually exclusive nature of $\{E_{-(L-1)}\dots E_0\}$ and $\{E_1,\dots E_{(L-1)}\}$ respectively. Now consider second term from the above.
	\begin{align} 
	\sum_{i=1}^{L-1}&\sum_{j=1}^i\mathbb{P}(E_i\cap E_{-(L-j)}) \nonumber \\
	=&\sum_{i=1}^{L-1}\sum_{j=1}^i\mathbb{P}(E_{-(L-j)})\mathbb{P}(E_i| E_{-(L-j)}) \nonumber  \\
	= &\frac{1}{KL+\mathbb{E}[W]}\sum_{i=1}^{L-1}\sum_{j=1}^i\mathbb{P}(W=(i-j)K) \nonumber \\
	= &\frac{1}{KL+\mathbb{E}[W]}\sum_{i=1}^{L-1}(L-i)\mathbb{P}(W=(i-1)K)\nonumber. 
	\end{align}
	In the second equality follows from \eqref{eq:pi} and given $E_{-(L-j)}$, $E_i$ takes place if backoff duration~($W$) is exactly $(i-j)K$ slots. Using above in \eqref{eq:prob_inter_1} gives
	\begin{align}\label{eq:P_I2}
	p_{I2}
	=&\frac{2L-1}{KL+\mathbb{E}[W]} \nonumber \\
	&-\frac{1}{KL+\mathbb{E}[W]}\sum_{i=1}^{L-1}(L-i)\mathbb{P}(W=(i-1)K).
	\end{align}
	The probability of tagged radar suffering no interference in two radar network is given by $1-p_{I2}$, and in $M$ independent radar network is given by $(1-p_{I2})^{M-1}$. So, the probability of interference at the tagged radar in $M$ radar network is,
	\[p_I = 1-(1-p_{I2})^{M-1}.\]
\end{IEEEproof}
\begin{remark}
If the distance between the tagged and interfering radars is $d=j\Delta c, j \in \{0,1, \dots\}$, we have the same expression for $p_{I2}$.
\end{remark}
\subsection{The case $\delta > 0, d \geq 0$}
Consider the two radar system with $\delta \in [0, \Delta), d > 0$. The time condition for interference~\eqref{eq:time_condition} at the tagged radar in this scenario reduces to 
\[(t_1-1)\Delta < t_2 \Delta + \delta + \frac{d}{c}<(t_1+1)\Delta, \]
where $t_1, t_2$ are chirp start time slot indices of tagged and interfering radars, $ d $ is distance and $\delta$ is the clock offset at interfering radar with respect to the tagged radar. Now consider above time condition, and let $\delta_1 =\delta +\frac{d}{c}$.
\begin{align*}
(t_1-1)\Delta <& t_2 \Delta +\delta_1 <(t_1+1)\Delta \nonumber \\
(t_1-1)-\frac{\delta_1}{\Delta} <& t_2   <(t_1+1)-\frac{\delta_1}{\Delta}.
\end{align*}
 As $t_1, t_2$ are integers, 
 \begin{align}\label{eq:t_2}
 t_2 = t_1-1-\Big\lfloor\frac{\delta_1}{\Delta}\Big\rfloor \text{  or } 
t_1+1-\Big\lceil\frac{\delta_1}{\Delta}\Big\rceil.
 \end{align}
 Observe that irrespective of the exact values of $d, \delta$ chirp start slot index~$t_2$ can take two successive integers values. If the chirp at the interfering radar starts at any of these two slot indices then tagged radar suffers interference. Note that in the first scenario $d=0, \delta = 0 $ (Section~\ref{subsec:1_prop}) the number of possible slot indices that can cause interference was only one.
 \par Now we define events similar to the one defined in Section~\ref{subsec:1_prop} that can cause interference at the tagged radar. Without loss of generality let us assume a packet transmission started at the tagged radar at time slot $t_1$. Let us define events $\{E_{-(L-1)}, \dots,E_{L-1}\}$ at the interfering radar, where event $E_i, i \in \{-(L-1)\dots L-1\}$ is starting packet transmission at slot
  \[t_1+iK-1-\Big\lfloor\frac{\delta_1}{\Delta}\Big\rfloor \text{  or  }t_1+iK+1-\Big\lceil\frac{\delta_1}{\Delta}\Big\rceil.\] These are the only events that can cause interference at the tagged radar. 
  \par For example consider $L=2, K=2, \delta = 0.5\Delta, d=0.2c\Delta,\text{ so }\delta_1=0.7\Delta$. Consider $E_{0}$, transmission cycle at interfering starts at slot $t_2 = t_1-1$ or $t_2 = t_1$. Similarly for $E_1$, the transmission cycle starts at  $t_2=t_1+1\text{ or } t_1+2 $ as explained in Figure~\ref{fig:2_prop}.
\begin{figure}	
  	\centering
  	\includegraphics[width=7cm, height=4.5cm]{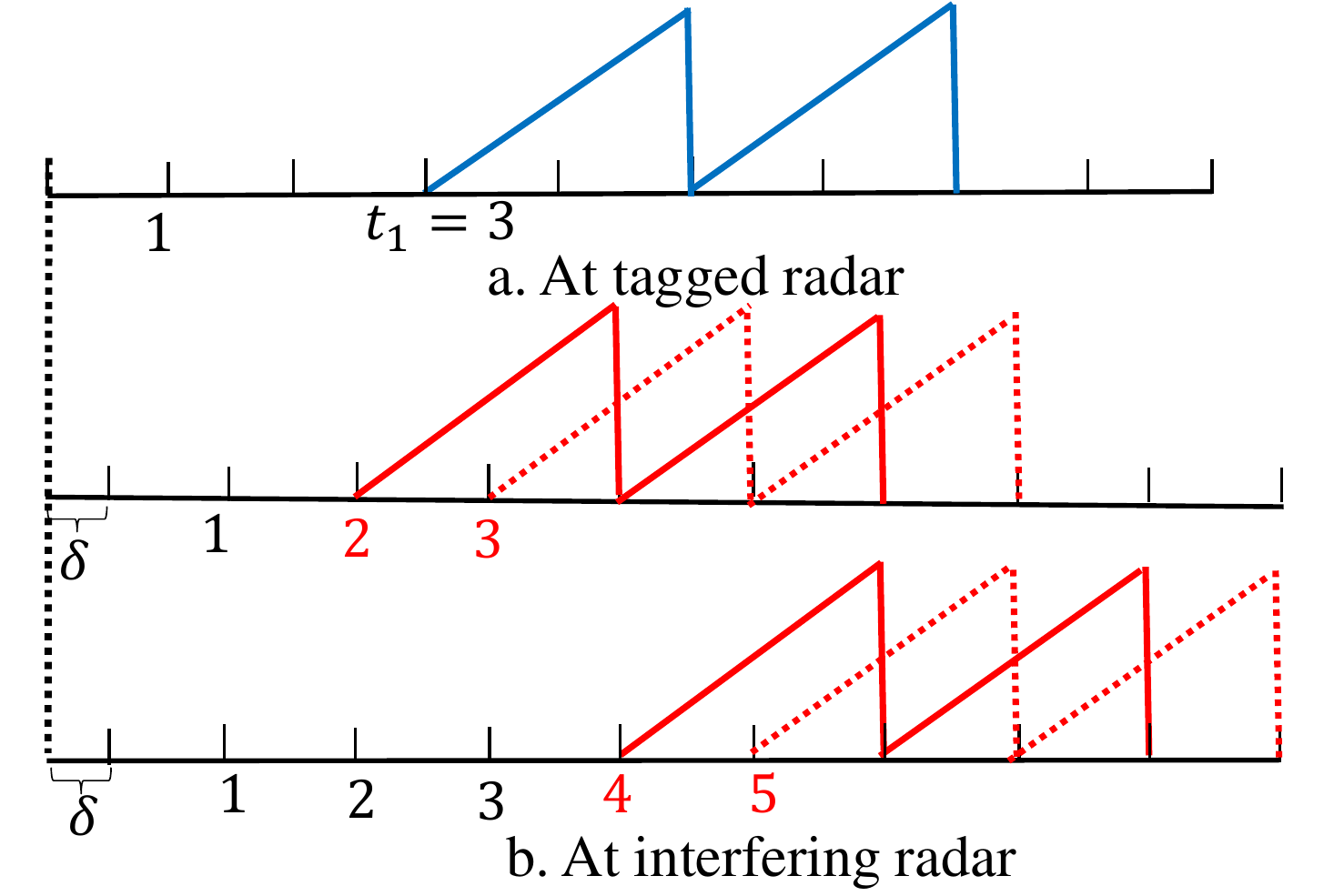}
  	\caption{Illustrating the interference events $E_0, E_1$ for $K=2,L=2$, a. The tagged radar transmitting a packet, b. Interfering events at the interfering radar }
  	\label{fig:2_prop}
\end{figure} 
 We compute the probability of these events in Lemma \ref{lemma:prob_E_i}.
\begin{lemma}\label{lemma:prob_E_i}
The probability of event $E_i, i \in \{-(L-1)\dots L-1\}$ is given by
\begin{align}\label{eq:PI_2_case_b}
\mathbb{P}(E_i) = \begin{cases}
\frac{2}{KL+\E[W]} \text{     for } KL>1 \\
\frac{1}{1+\E[W]}  +\frac{1-\P(W=0)}{L+\E[W]}\text{ for } KL=1.
\end{cases}
\end{align}
\end{lemma}
\begin{IEEEproof}
First we consider the case $KL=1$. In this case we have only one chirp of length one slot in each transmission cycle. Let $t_1$ be the start of packet transmission at the tagged radar. From \eqref{eq:t_2} we have two successive slots in which if the packet transmission starts at the interfering radar the tagged radar suffers interference. Let us call these $t_2, t_2'$. Observe that when $KL=1$, two consecutive packets can start in $t_2,t_2'$ if the backoff duration is zero between them as shown in Figure~\ref{fig:KL=1}b. Note that if $KL>1$ two consecutive packets can not start in any two successive slots. The event $E_i$ is defined as the packet starting in any one of these slots.
\begin{figure}	
	\centering
	\includegraphics[width=4.5cm, height=2.75cm]{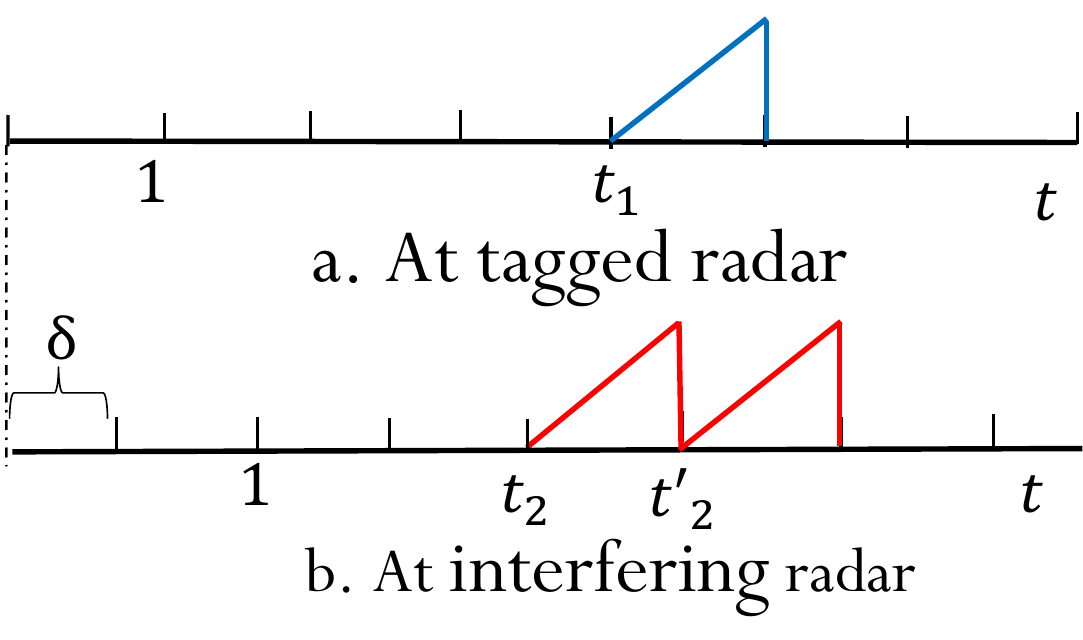}
	\caption{
		Illustrating the interference event $E_0$ for $KL=1$, a. The tagged radar transmitting a packet, b. Interfering events at the interfering radar }
	\label{fig:KL=1}
\end{figure}
\par In order to compute $\P(E_i)$ let us define the state of a radar in a slot as follows. We say that the radar is in state $0$ if it is transmitting. Further, we say that it is in state $j$ if it is in $j$-th consecutive backoff slot. Note that the latter case arises only if the radar is in the middle of a backoff of length greater than or equal to $j$. Also, notice that these definitions of states are different from those in Section~\ref{sec:Performance_analysis}. Let us define $\alpha_j$ to be the long term fraction of slots the radar spends in state $j$. The radar's state evolves according to a renewal process with the slot in which the radar is in state $0$ being the renewal epochs. We can therefore derive $\alpha_j, \text{ for  }j \geq 1$s as follows.

\begin{align*}
\alpha_j 
&=\lim_{T\rightarrow \infty} \frac{\text{No. of transmission cycles until slot } T\text{ }}{T}\P(W\geq j)\nonumber \\
&=\frac{\P(W\geq j)}{1+\E[W]},
\end{align*}
the last equality follows from Elementary Renewal theorem~\cite{AK}.
\par Now consider the event $E_i$ at interfering radar i.e., starting the chirp at slot $t_2$ or radar is in the last time slot of the backoff at slot $t_2$, (this automatically makes $t'_2=t_2+1$ chirp starting slot). Probability of starting chirp in $t_2$ is given by $\frac{1}{1+\E[W]}$. Now the probability of $t_2$ being the last slot of backoff is given by 
\[\sum_{j \geq 1}\alpha_j \frac{\P(W=j)}{1-F(j-1)}.\] Using these two we compute $\P(E_i)$.
\begin{align*}
\P(E_i) =& \P(\text{transmission cycle starting at slot }t_2 \text{ or }t_2+1)  \\
=&\frac{1}{1+\E[W]}+\sum_{i \geq 1}\alpha_j \frac{\P(W=j)}{1-F(j-1)} \\
=&\frac{1}{1+\E[W]}+\frac{1-\P(W=0)}{1+\E[W]}
\end{align*}
\par Now consider $ KL > 1 $. In this case the packet transmission starting slots at $ t_2 $ and $t_2 +1 $ are mutually exclusive events and from \eqref{eq:pi} we write $\P(E_i)=\frac{2}{KL+E[W]}$.
\end{IEEEproof}
 \par Using the $\P(E_i), i \in \{-(L-1)\dots L-1\}$ derived in Lemma \ref{lemma:prob_E_i} we derive the probability of interference at the tagged radar in two radar network and $M$ radar network in Proposition~\ref{prop:Prop_2}.
\begin{proposition}\label{prop:Prop_2}
	\begin{enumerate}
		\item The probability of interference at the tagged radar in two radar network is given by,
		\begin{align*}
		p_{I2} = &(2L-1)\mathbb{P}(E_1) \nonumber \\ 
		&-\frac{1}{KL+\mathbb{E}[W]}\sum_{i=0}^{L-2}(L-i-1)\Big( 2\P\big(W=iK\big) \nonumber \\&+\P\big(W=iK+1\big)+\P\big(W=iK-1\big)\Big).
		\end{align*}
		where $\mathbb{P}(E_i)$ is given by \eqref{eq:PI_2_case_b}.
		\item The probability of interference at the tagged radar in $M$ radar system is given by,
		\begin{align*}
		p_{I} = 1-(1-p_{I2})^{M-1}.
		\end{align*}
	\end{enumerate}
\end{proposition}
\begin{proof}
	Using similar arguments as in the proof of Proposition~\ref{prop:Prop_1} we have,
	\begin{align}\label{eq:prob_inter_1_case_b}
	p_{I2}=&\mathbb{P}\big(\cup_{i=-(L-1)}^{L-1}E_i\big) \nonumber \\ 
	=&\sum_{i=-(L-1)}^{L-1}\mathbb{P}(E_i)-\sum_{i=1}^{L-1}\sum_{j=1}^i\mathbb{P}(E_i\cap E_{-(L-j)}).
	\end{align}
	Now consider the second term from above,
	\begin{align*}\label{eq:prob_inter_2_case_b}
	\mathbb{P}&(E_i\cap E_{-(L-j)})  \nonumber \\
	=& \mathbb{P}(E_{-(L-j)})\mathbb{P}(E_i| E_{-(L-j)})\nonumber \\
	=&\frac{1}{KL+\mathbb{E}[W]}\Big(2\mathbb{P}\big(W=(i-j)K\big)\nonumber \\
	&+\mathbb{P}\big(W=(i-j)K-1\big)+\mathbb{P}\big(W=(i-j)K+1\big) 
	\end{align*}
	 Using above in \eqref{eq:prob_inter_1_case_b}, we get $p_{I2}$.

	Using the similar arguments as in the proof of Proposition~\ref{prop:Prop_1} we have,
	\[p_I = 1-(1-p_{I2})^{M-1}.\]
\end{proof}

\subsection{Throughput}\label{subsec:throughput}
We compute the throughput we defined in Section~\ref{subsec:Radar Network} in the following Proposition~\ref{prop:theta}.
\begin{proposition}\label{prop:theta}
	Throughput of $M$ radar network is given by,
	\begin{align}\label{eq:Theta}
	\Theta = MLK\pi_1(1-\mathbb{P}_{I2})^{M-1}
	\end{align}
\end{proposition}
\begin{IEEEproof}
	From the definition of throughput~(see \ref{para:Theta}), we have
	\[\Theta = \lim_{T\rightarrow \infty} KL\frac{S(T)}{T},\]
where $S(T)$ total number of packets transmitted successfully by all radars in $T$ time slots. Now consider a radar in the network. Let the total number of transmissions cycles until time slot $T$ is given by $\phi(T)$. The average number of successful packets transmission cycles is given by $\phi(T)(1-p_{I2})^{M-1}$ and number of chirps transmitted is given by $L\phi(T)(1-p_{I2})^{M-1}$. In $M$ independent radars network the average number of successful transmission cycles is given by $M\phi(T)(1-p_{I2})^{M-1}$. Using this we can write,
\begin{align*}
\Theta &= \lim_{T\rightarrow \infty} \frac{KML\phi(T)(1-p_{I2})^{M-1}}{T} \nonumber \\
&=\frac{1}{KL+\E[W]}MKL(1-p_{I2})^{M-1} \nonumber \\
&=MLK\pi_1(1-p_{I2})^{M-1} 
\end{align*}
where the first and second equalities follow from Elementary Renewal and Renewal Reward theorems~\cite{AK}.
\end{IEEEproof}

\subsection{$p_I$ and $p_{opt}$ for geometric backoff } \label{rem:p_opt}
In this section we present results obtained for geometric backoff distribution, $W\sim\text{Geo}(p)$ where $p$ is the attempt probability and supported on the set $\{0,1,2,\dots\}$.
\begin{enumerate}
	\item The probability of interference at the tagged radar in two radar network is given by, 	
	\begin{align*}
	p_{I2}=& (2L-1)\mathbb{P}(E_i) \nonumber \\ 
	&-\frac{p^2}{pKL+1-p}\sum_{i=1}^{L-1}(L-i)\Big(2(1-p)^{(i-1)K} \nonumber\\
	&+(1-p)^{(i-1)K+1}+(1-p)^{(i-1)K-1}\Big), 
	\end{align*}
	where \begin{align*}
	\mathbb{P}(E_i) = \begin{cases}
	\frac{2}{KL+\E[W]} \text{     for } KL>1 \\
	\frac{1}{1+\E[W]} + \frac{p\E[W]}{1+\E[W]}\text{ for } KL=1.
	\end{cases}
	\end{align*}
	and $\E[W]=\frac{1-p}{p}$.
	\item  
	Further, we consider number of chirps per packet equal to one i.e., $L=1$ and evaluate the attempt probability $p_{opt}$ at which the network has the maximum throughput.	
\begin{align*}
p_{opt} = \begin{cases}
\frac{1}{2M+1-K}, \text{ if  } 2M+1-K > 0 \\
1, \text{   otherwise. }  
\end{cases}
\end{align*}
\end{enumerate}
We observe that a centralized transmission scheduling can achieve a throughput close to $ \min\{K,M\} $ which is much more than the throughput achieved by ALOHA, e.g., see Figure~\ref{fig:ALOHA_1}b. This motivates us to explore other MAC protocols.
\section{CSMA} \label{sec:CSMA Algo}
\par In slotted ALOHA, the radars initiate transmissions irrespective of the activities of other radars. Consequently, in most of the scenarios, slotted ALOHA achieves a throughput far below the best possible values~(see Figure~\ref{fig:ALOHA_1}b). This motivates us exploring other MAC protocols  that account for other radars' activities. In this section, we study a MAC protocol in which a radar senses the medium following each backoff, and initiates a transmission only if the medium is sensed idle.

\paragraph*{CSMA} In CSMA, a radar after each backoff, senses the medium for one slot which is referred to as clear channel assessment (CCA) slot. In the CCA slot, the radar sets the Tx chirp frequency input to the mixer to $f_{\min}$ and observes the output of its LPF at the end of the CCA slot. If the radar observes a signal (of frequency less than $f_H$), it assumes the CCA to have failed and enters into another random backoff. If the radar does not observe a signal at the end of the CCA slot, it assumes the CCA to have succeeded and starts a packet transmission on completion of the CCA slot. We illustrate the flow chart of our CSMA protocol in Figure~\ref{Algo:CSMA Algorithm}. 
\begin{remark}
	Note that if the radar observes a signal at the end of the CCA slot and still starts a packet transmission in the next slot, this packet would interfere with the packet that caused the observed signal. In other words a transmission following a CCA failure would certainly lead to interference. On the other hand, a transmission following a CCA success need not be interference free as we explain in Section~\ref{subsubsec:coll-L>1}. 
\end{remark} 
\par Recall that in ALOHA each radar transmits packets irrespective of other radars' transmissions. But in CSMA all the radars' backoffs, CCAs and packet transmissions are coupled. In the following subsections we analyze CSMA in various cases.

\begin{figure}	
	\centering
	\includegraphics[width=5cm, height=5cm]{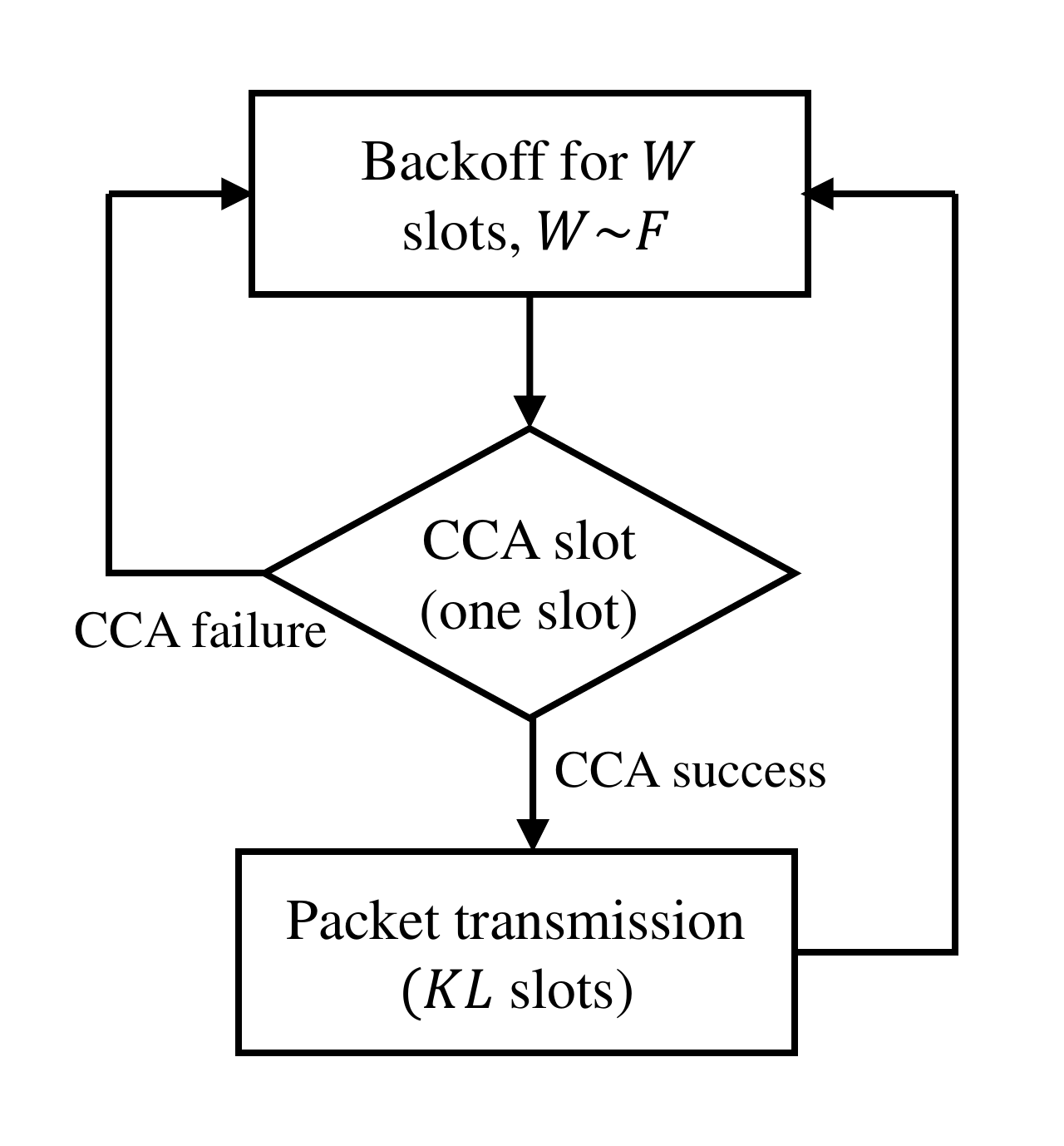}
	\caption{CSMA, After every backoff radar senses the medium for $f_{\min}$~(CCA), if CCA is successful then it transmits else enter into backoff}
	\label{Algo:CSMA Algorithm}
\end{figure}

\subsection{Collocated radars} \label{subsec:collocated-Async}
We consider a collocated and asynchronous radar network. Further, we first consider number of chirps per packet, $L$ to be one. We consider $L > 1$ subsequently.

\subsubsection{$L=1$} \label{subsubsec:collocated-L=1}
 In this case we observe that the radars do not suffer interference in CSMA. Observe that, in CSMA, for a CCA success at a radar, the immediately preceding chirp transmission by any of the other radars must have started at least $\Delta$ time before the end of the CCA slot. Consequently, start times of packet transmissions by the radars are separated by at least $ \Delta $, and so the packets will not interfere with each other. We illustrate it by showing sample CCA failure and CCA success events in Figure~\ref{fig:CSMA-collocated-L=1}. Let $t_1$ be the CCA slot at the tagged radar. The tagged radar senses the $f_{\min}$ of the chirp transmitted (dotted red color) by the interfering radar, experiences CCA failure as shown in Figure~\ref{fig:CSMA-collocated-L=1}. The next CCA attempt is successful and tagged radar starts transmission immediately after the CCA slot.
 \begin{figure}	
 	\centering
 	\includegraphics[width=5.75cm, height=1.5cm]{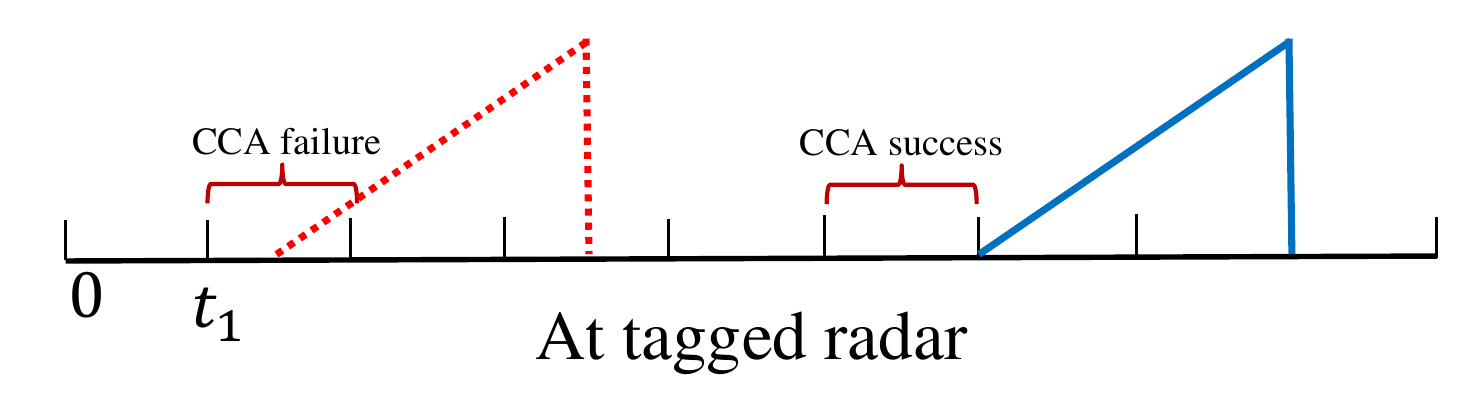}
 	\caption{Illustrating CCA success, CCA failure and packet transmission at the tagged radar in a collocated radar network for $L=1, K=2$}
 	\label{fig:CSMA-collocated-L=1}	
 \end{figure}
\par In the following we compute the throughput. Let us consider a tagged radar. We assume that the backoffs are sampled from Geo$(p)$ distribution. Let $ I_s(T) $ be the number of successful transmissions at the tagged radar until slot $ T $. Further, let $p_s$ be the conditional CCA success probability at the tagged radar i.e., the probability of CCA success given the tagged radar is in a CCA. We define the tagged radar's transmission rate and the aggregate throughput, $ p_r $ and $\Theta $, respectively, as follows.
\begin{align}\label{eq:p_r}
p_r := 
\lim_{T\rightarrow \infty}\frac{I_s(T)}{T}.
\end{align}
 Using above and $L=1$ in the definition of throughput~\ref{para:Theta} we have
\begin{align}\label{eq:throghput_Async}
\Theta = MKp_r.
\end{align} 
\begin{remark}\label{rem:max_pr}
	Observe that the maximum throughput that can be achieved by any algorithm for $L=1$ is $\min\{M,K\}$. Using this in \eqref{eq:throghput_Async} gives $p_r\leq \frac{1}{\max\{M,K\}}$. 
\end{remark}
In the following we compute the throughput using decoupling approximation and fixed point analysis~\cite{AnuragKumar},\cite{Bianchi}. 
Towards this we find the relation between $p_r, p_s$ in Lemma~\ref{lemma:p_r}.
\begin{lemma}\label{lemma:p_r}
	\begin{align*}
		p_r = \frac{1}{K+\frac{1}{pp_s}}
	\end{align*}
\end{lemma}
\begin{IEEEproof}
Let us define the state of a radar in a slot as follows. We say that the radar is in state $-1$ if it is sensing the medium, state $0$ if it is in backoff. Further, we say that radar is in state $i, 1\leq i \leq K$ if it is transmitting $i$-th slot of the packet. The radar's state evolves according to a renewal process with the slots in which the radar state is $1$ being the renewal
epochs. Let the duration between any two renewal epochs $i$ and $i+1$ is given by $D_i$. Observe that $D_i, i\in\{1,2,\dots\} $ is random and independent across $i$. Further, $\E[D_i] = K+\frac{1}{pp_s}$. 
\par Using Elementary Renewal theorem \cite{AK},
\begin{align*}
 p_r = &\lim_{T\rightarrow \infty}\frac{I_s(T)}{T} \\
     = & \frac{1}{K+\frac{1}{pp_s}}
\end{align*}
almost surely.
\end{IEEEproof}
To compute $p_s$ we use the decoupling approximation \cite{AnuragKumar}, \cite{Bianchi}. The decoupling approximation assumes that the aggregate attempt process of the other radars is independent of the backoff process of the tagged radar. Recall in any $\Delta$ interval of tagged radar's timeline, at most one radar can start packet transmission. The unconditional probability of any one among $M-1$ radars starting transmission in one slot is given by $(M-1)p_r$. The unconditional probability of no radar starting transmission is given by $1-(M-1)p_r$. From Remark~\ref{rem:max_pr} we observe that $1-(M-1)p_r > 0$. We want to calculate the conditional probability, $p_s$ at the tagged radar i.e., given the tagged radar is in a CCA slot the probability of its CCA success. But we use the unconditional CCA success probability $1-(M-1)p_r$ for $p_s$ in Lemma~\ref{lemma:p_r} to compute $p_r$. Further, we substitute $p_r$ in \eqref{eq:throghput_Async} to compute the throughput. 
\begin{align*}
p_r & = \frac{1}{K+\frac{1}{p\big(1-(M-1)p_r\big)}} 
\end{align*}
For a given $M,p$ we can write above as
\begin{align*}
p_r = f(p_r).
\end{align*}
In Lemma~\ref{lemma:fixed-point} we show that $p_r$ can be computed using fixed point analysis. 
\begin{lemma}\label{lemma:fixed-point}
 For a $p_r = f(p_r)$ has a unique fixed point.
\end{lemma}
\begin{IEEEproof}
\begin{itemize}
	\item [a.]$f(.)$ is a continuous map from $[0,\frac{1}{\max\{M,K\}}]$ to $[f(0),f\big(\frac{1}{\max\{M,K\}}\big)]$ and decreasing in $p_r$
	\item [b.] $f\big(\frac{1}{\max\{M,K\}}\big)<\frac{1}{\max\{M,K\}}$
\end{itemize}

Hence by intermediate value theorem there exists a unique fixed point. 
\end{IEEEproof}
\par We use $p_r$ obtained from Lemma~\ref{lemma:fixed-point} in \eqref{eq:throghput_Async} to compute the throughput. We find from Section~\ref{sec:numerical_res}~(see Figure~\ref{fig:fixexpt}b) that the throughput obtained from the fixed point analysis matches the simulation with less than $10$ percent error for $M>K$. Further, from Figure~\ref{fig:Async_Col_L1_Theta_vs_p} we observe that in CSMA the throughput increases in attempt probability, and for all values of $p$ throughput of CSMA is better than ALOHA.

\subsubsection{$L>1$}\label{subsubsec:coll-L>1}
We notice that when $L>1$, packets may interfere under the proposed CSMA. We illustrate this using Figure~\ref{fig:CSMA-collocated-L=2}. Let us assume that the tagged radar starts CCA at slot $t_1$, while other radar transmitting a packet~(dotted red color). The tagged radar CCA is successful as the starting time of other radar is earlier than $t_1$. However the subsequent packet~(solid blue color) transmission of the tagged radar is interfered by the the second chirp of the packet from the other radar. Observe that during interference one of the radar should be transmitting the first chirp in a packet and other must be transmitting chirp~$j,1< j\leq L$ in a packet.
\begin{figure}	
	\centering
	\includegraphics[width=5.5cm, height=2cm]{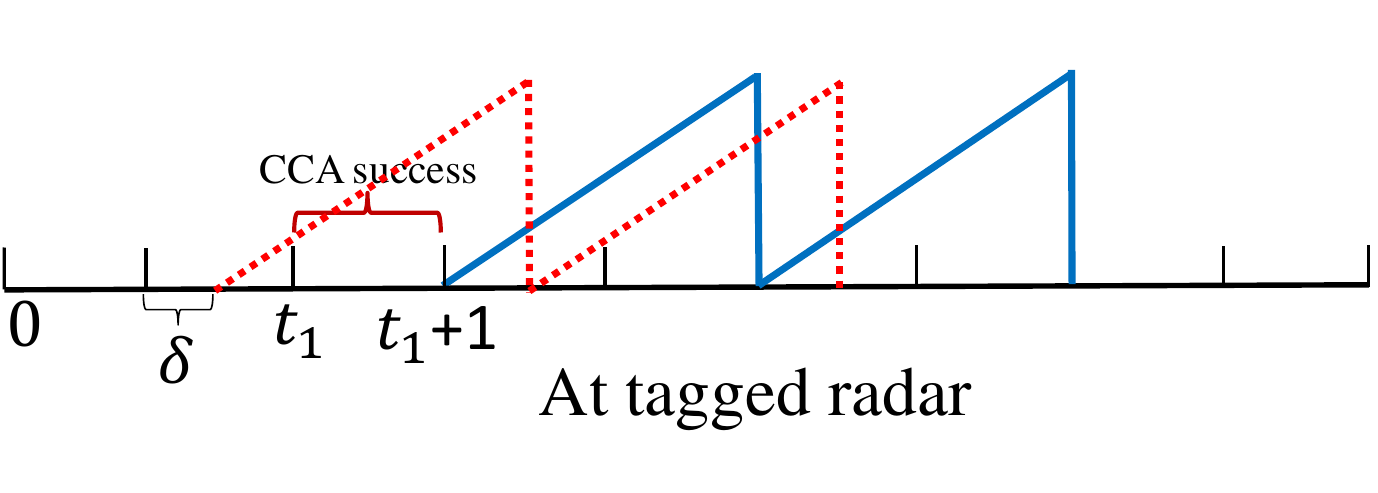}
	\caption{Illustrating interference at the tagged radar after CCA success in a collocated radar network for $L=2, K=2$}
	\label{fig:CSMA-collocated-L=2}	
\end{figure}
\par From simulation (see~Figure~\ref{fig:Async_Col_L4_Theta_vs_p}) we observe that CSMA's throughput is higher than that of slotted ALOHA for all values of $p$. We also observe that the throughput is a non-monotonic function of $p$.
\subsection{Non-collocated radars}\label{subsec:csma-non-coll} 
We consider a non-collocated and asynchronous radar network. In Section~\ref{subsubsec:collocated-L=1} we saw that in a collocated radar network radars do not suffer interference for $L=1$. But in non-collocated setting radars suffer interference for $L=1$ also. We explain this in the following using Figure~\ref{fig:CSMA-non-collocated-L=1}. Let $d$ be the distance between the tagged and interfering radars and $d<d_{I,\max}$. The tagged radar receives the packets transmitted by the interfering radar with a delay of $\frac{d}{c}$. Let $d = c\Delta$. This causes a delay of one slot as shown in Figure~\ref{fig:CSMA-non-collocated-L=1}. Let us assume that the tagged radar started CCA at slot $t_1$. In Figure~\ref{fig:CSMA-non-collocated-L=1} observe that though the interfering radar started a packet transmission during the CCA slot of the tagged radar, the tagged radar can not sense it because of the delay caused by the distance. The packet received at the tagged radar is shown in dotted lines (see Figure~\ref{fig:CSMA-non-collocated-L=1}a). The packet transmitted by the interfering radar reached the tagged radar after the CCA slot. So, the CCA was successful and the tagged radar started the packet transmission at slot $t_1+1$. The tagged radar received the packet during the first slot of packet transmission and suffered interference. 
\par Now consider another CCA slot at $t_2$. The tagged radar senses the packet from the other radar, it is a CCA failure and packet transmission did not start at slot $t_2+1.$
\begin{figure}	
	\centering
	\includegraphics[width=6cm, height=3.5cm]{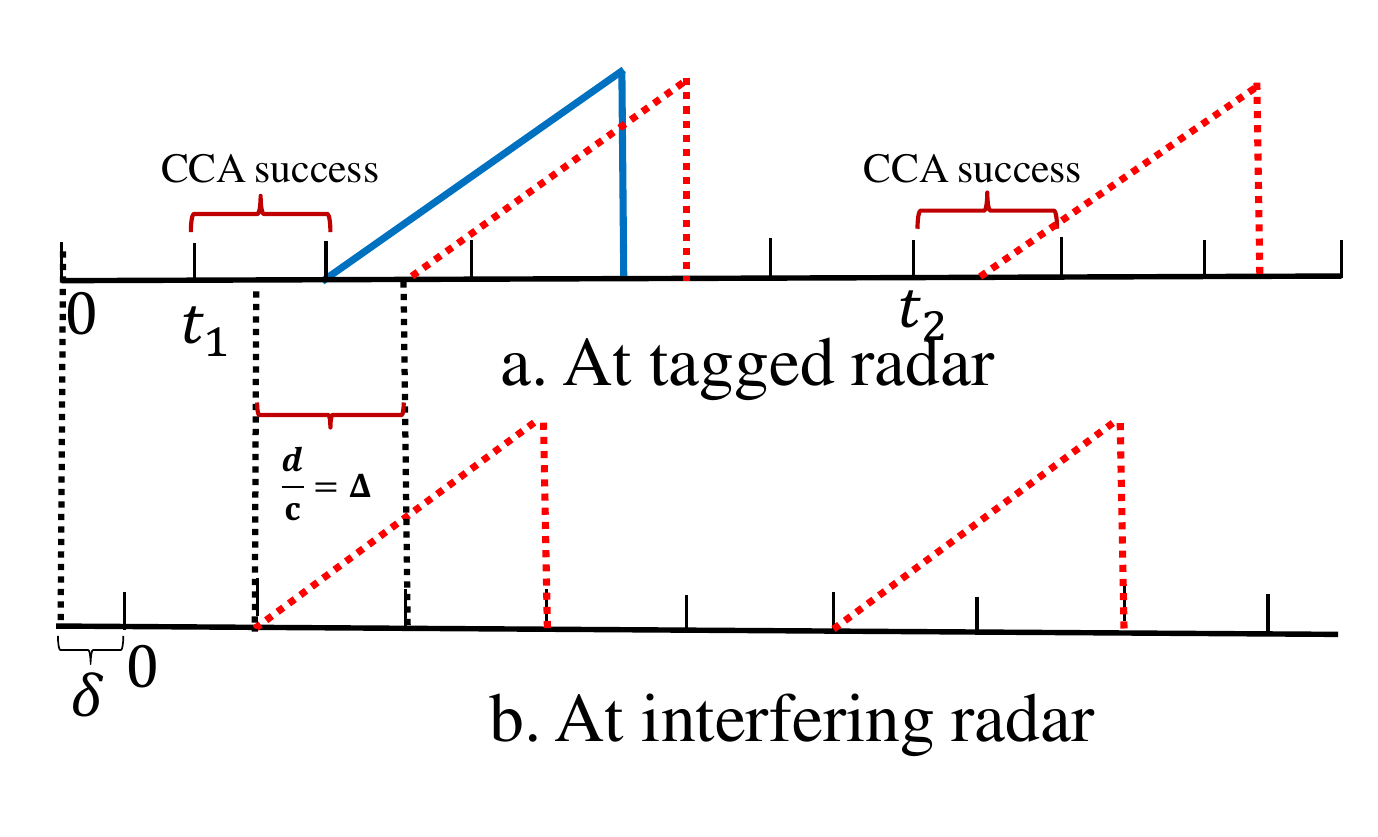}
	\caption{Illustrating interference in a non-collocated radar network for $L=1, K=2$, a. Transmitted and received packets at the tagged radar, b. Transmitted packets at the tagged radar}
	\label{fig:CSMA-non-collocated-L=1}	
\end{figure}
So, the CSMA avoids interference caused by packets that are received at the tagged radar in the CCA slot. 
From Section~\ref{sec:numerical_res} (see~Figure~\ref{fig:Non_Col_L1_Theta_vs_p}) we observe that for $L=1$ in CSMA the throughput increases in $p$, and for all values of $p$ throughput of CSMA is better than ALOHA.
\section{Numerical Results}\label{sec:numerical_res}
\par In this section, we simulate an FMCW radar network with $M$ radars. The radars' clocks are not synchronized.
Assume that there is a hypothetical global clock, the clock offsets for each radar is assumed to be uniformly distributed random variable in $[0, \Delta)$. More precisely,
we consider the cases of collocated and non-collocated radars. In the case of non-collocated radars, we assume that the radars are placed on a line segment $[0, d_{I,\max}]$. Radar locations are chosen from i.i.d. $ \text{Unifrom}[0, d_{I,\max}] $ distributions. We set $ d_{I,\max} = 4d_{\max} $. We assume that backoffs are sampled from Geo($ p $) distribution. 
\par We first verify the probability of interference~($p_I$) and throughput~($\Theta$) obtained in Proposition~\ref{prop:Prop_2} and Proposition~\ref{prop:theta}, respectively for different values of packet length~$(L) $, attempt probability~$(p)$ and number of radars~$(M)$ in ALOHA. Then we verify the fixed point solution obtained in Lemma~\ref{lemma:fixed-point} for collocated radar networks by varying $L,p,M$ in CSMA. Finally, we compare the performance of ALOHA and CSMA by varying $L,p,M$.
\par In Figure~\ref{fig:ALOHA_1} we demonstrate how $p_I$ and $\Theta$ vary with the attempt probability~$p$ and $L$ in ALOHA. Towards this we consider a network with $M=50$ radars with chirp length $K=40$. Further, we consider $L=1,4,5$.
\begin{figure}	
	\centering
	\includegraphics[width=8cm, height=3.65cm]{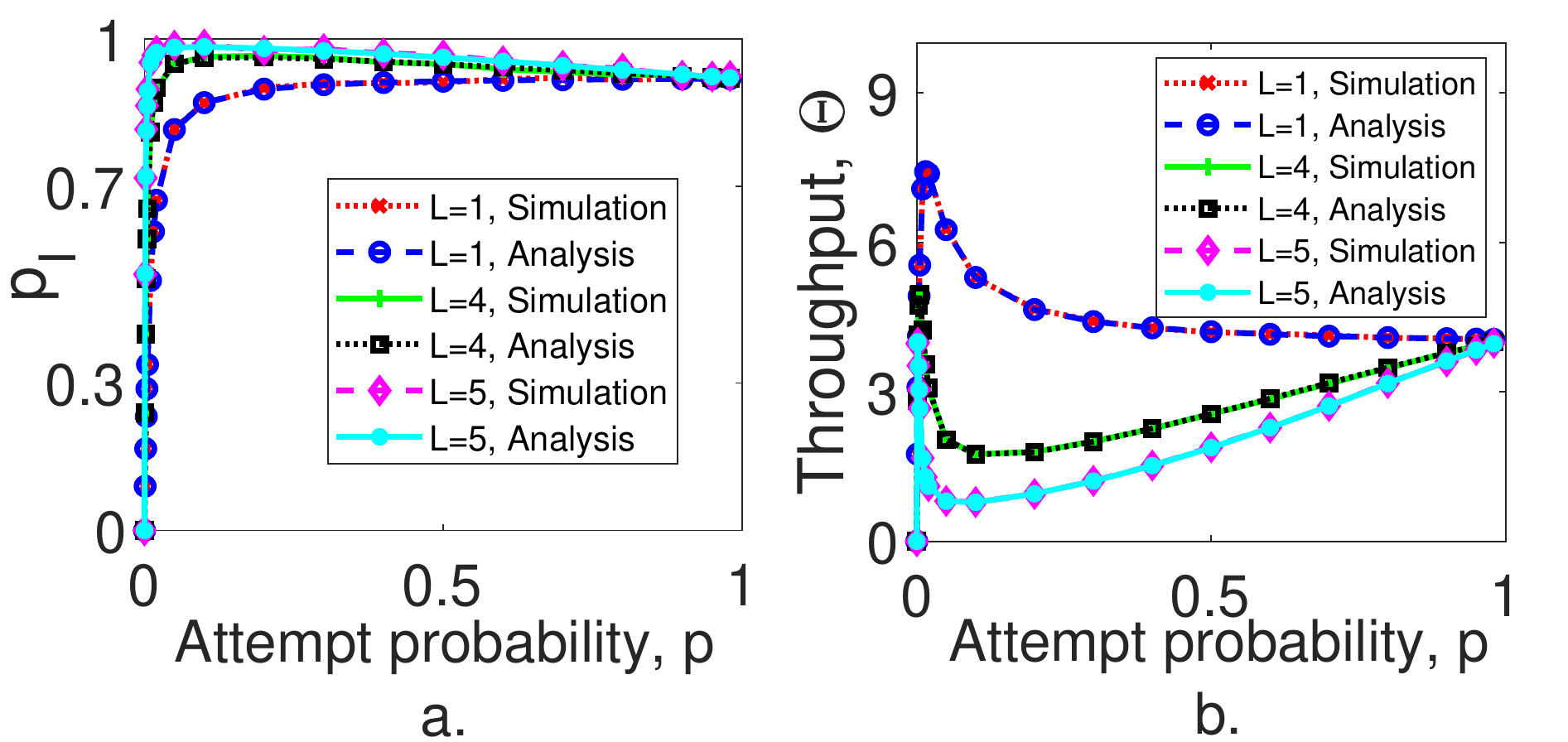}
	\caption{ Verifying the analytical results obtained in Propositions~\ref{prop:Prop_2} and \ref{prop:theta} for ALOHA,  a. Probability of interference vs attempt probability,   b. Throughput vs attempt probability for $M=50,K=40, L=1,4,5$}
	\label{fig:ALOHA_1}	
\end{figure}
\par We observe in Figure~\ref{fig:ALOHA_1} that $p_I, \Theta$ obtained from simulation and analysis match very closely~(less than $2$ percent error). Further, we also observe that at $p=1$, $p_I$ and $\Theta$ do not vary with $L$. This is because, at $p=1$, each radar generates a continuous train of packets with no backoff between them. We also observe that for $L>1$ interference probability increases and then decreases with $p$ which is different from wireless networks. Further, we find from Figure~\ref{fig:ALOHA_1}b that the maximum throughput for $L=1$ is achieved at $p_{opt} = \frac{1}{2M+1-K} = 0.016$~(see Section~\ref{rem:p_opt}).

\par In Figure~\ref{fig:fixexpt} we demonstrate the fixed point solution obtained in Lemma~\ref{lemma:fixed-point} for CSMA and compare it with the simulation. Towards this we consider a radar network with $M=50$ radars each using chirp length $K=40$ and $L=1$. In Figure~\ref{fig:fixexpt}a we plot fixed point solution for different attempt probabilities $p=0.2,0.4,0.6$. In Figure~\ref{fig:fixexpt}b we plot throughput obtained from fixed point analysis and the simulation by varying $p$ for $M=30,40,50$. We find that the error in throughput obtained from the fixed point analysis matches the simulation with less than $10$ percent error.
\begin{figure}	
	\centering
	\includegraphics[width=8cm, height=4cm]{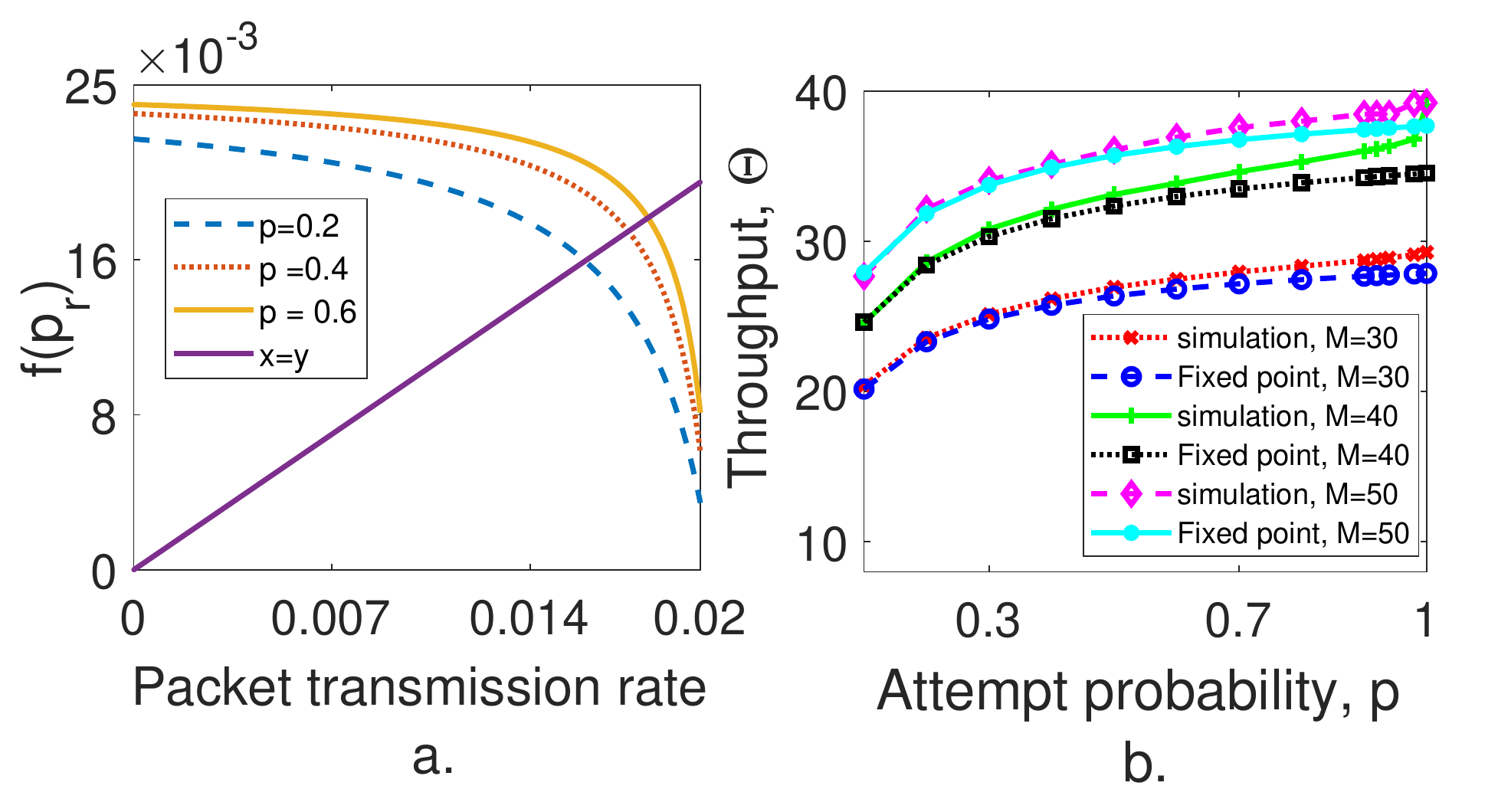}
	\caption{a. Fixed point analysis for $ M=50, K=40$, b. Throughput vs attempt probability for a collocated radar network $K=40, L=1$}	\label{fig:fixexpt}
\end{figure}

\par To compare the performance of ALOHA and CSMA we consider the following two network scenarios separately.
\begin{itemize}
	\item collocated and asynchronous radar network,
	\item non-collocated and asynchronous radar network. 
\end{itemize}

\subsection{Collocated and asynchronous radar network}\label{subsec:collocated-Async-sim}
In this section, we simulate a collocated radar network. We first consider the case where the number of chirps per packet $ L = 1 $. Subsequently, we consider $ L > 1 $ case.

\subsubsection {$ L =1 $}\label{subsub:coll-L=1}
\par To begin with we demonstrate how the throughput varies with $ p $. Towards this we fix $K=40$ and vary $M$ as shown in Figure~\ref{fig:Async_Col_L1_Theta_vs_p}. We observe that in CSMA the throughput increases in $p$. We observe that the network achieves maximum throughput at $p=1$. We explain this in the following. Recall from Section~\ref{subsubsec:collocated-L=1} that in a collocated and asynchronous network radars suffer no interference. At $p=1$ each radar continuously senses the medium until its CCA is successful i.e., the radar is either transmitting the packet or sensing the medium. Further, we observe that for $L=1$ the maximum throughput that can be achieved is $\min\{M,K\}$. We find that in Figure~\ref{fig:Async_Col_L1_Theta_vs_p} for $ M=30, 40, 50$ and $K=40$ CSMA achieves $29.3,39.1,39.25$ at $p=1$ respectively. 
In ALOHA the throughput decreases with the increase in the number of radars after maximum throughput. This is because each radar transmits packets irrespective of other radar packet transmissions. So, at a higher value of attempt probability every radar transmits packets often and suffers interference. Further, we observe that CSMA has better throughput than ALOHA. 
\begin{figure}	
	\centering
	\includegraphics[width=5.5cm, height=4cm]{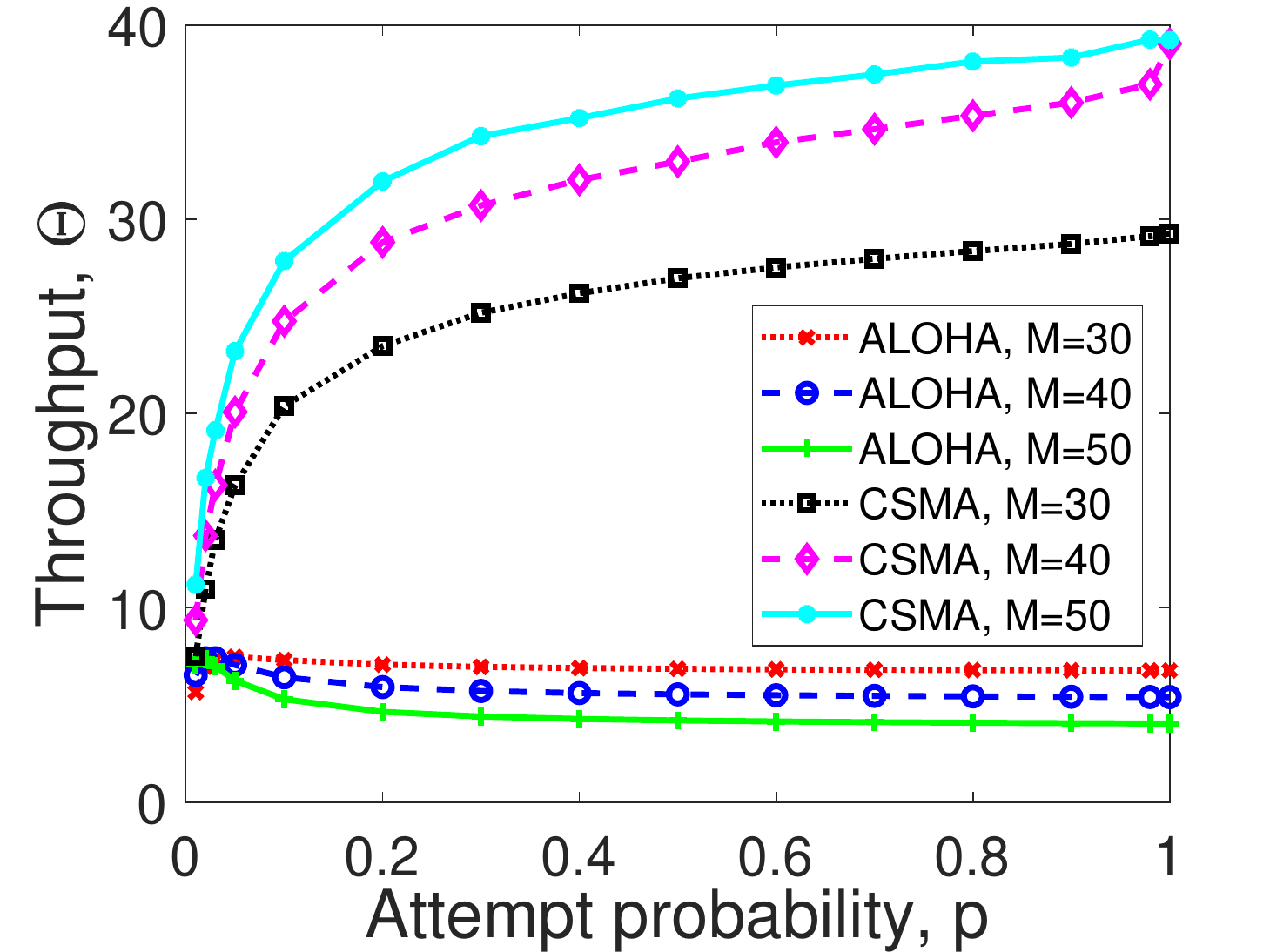}
	\caption{Throughput vs attempt probability for collocated and asynchronous radar network, $L=1, K=40, M = 30, 40, 50$}
	\label{fig:Async_Col_L1_Theta_vs_p}
\end{figure}
\par In Figure~\ref{fig:Async_Col_L1_OPT_Theta_N} we demonstrate how $p_{opt}$ and $\Theta_{opt}$ vary with $M$ for CSMA and ALOHA. Towards this, we fix $K=40$. From Figure~\ref{fig:Async_Col_L1_Theta_vs_p} we observed that in CSMA throughput increases in $p$, so observe $p_{opt}=1$ for all values of $M$ in Figure~\ref{fig:Async_Col_L1_OPT_Theta_N}a. Further, in Figure~\ref{fig:Async_Col_L1_OPT_Theta_N}b we observe that $\Theta_{opt}$ increases in $M$ in CSMA though the increase is small. In ALOHA we observe that $p_{opt}$ decreases with the increase in the number of radars. This is because for a fixed $p$, as the number of radars increase interference probability increases. So, radar has to transmit packets less often to avoid interference, i.e., $p$ has to decrease with the increase in $M$. From Figure~\ref{fig:Async_Col_L1_OPT_Theta_N}b we find that $\theta_{opt} $ decreases with $M$ in ALOHA though the decrease is small. Further, the optimum throughput achieved by CSMA is always more than the optimum throughput achieved by ALOHA.
\begin{figure}	
	\centering
	\includegraphics[width=8cm, height=3.7cm]{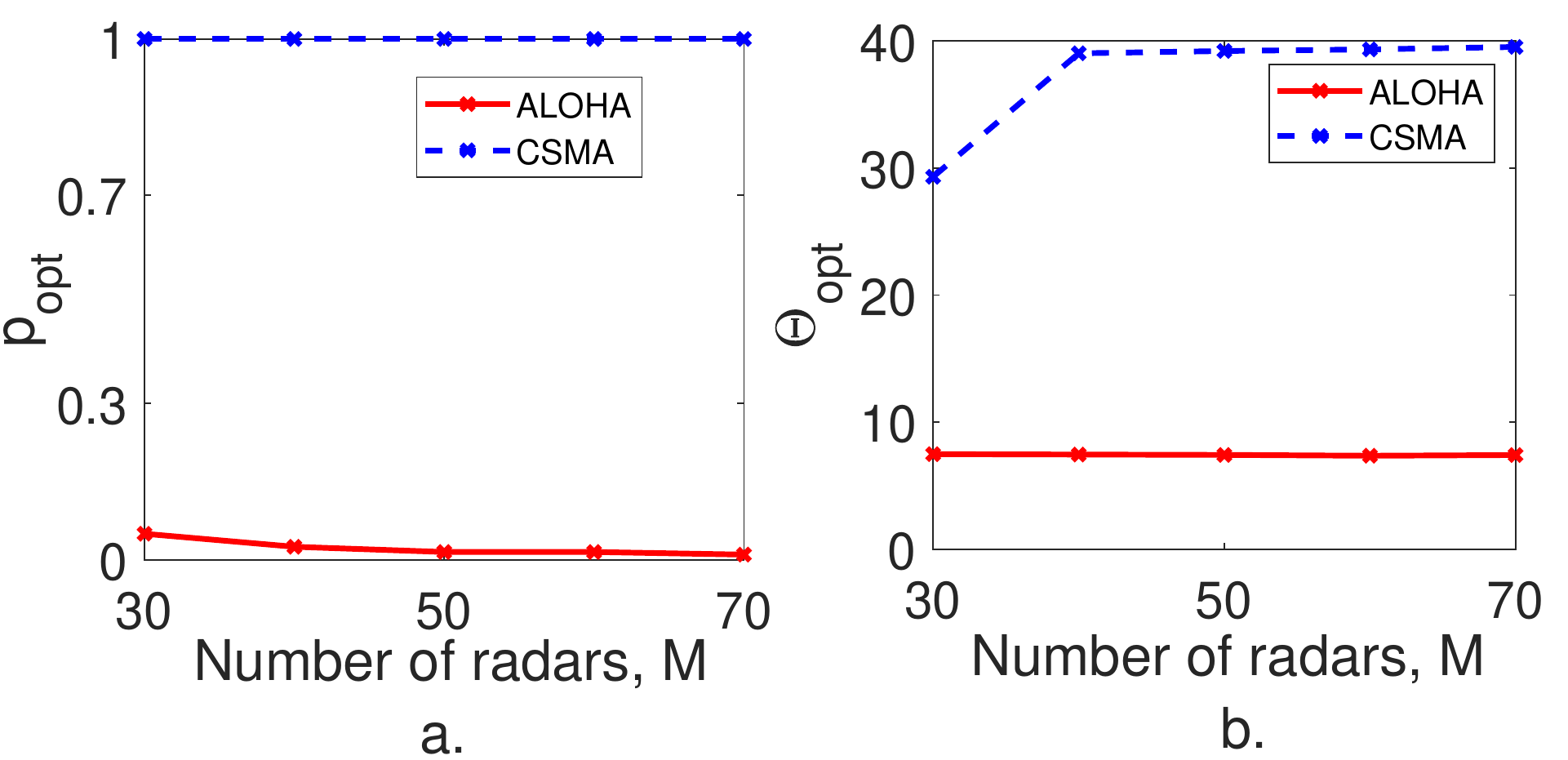}
	\caption{a. $ p_{opt}$ vs number of radars, b. $\Theta_{opt}$ vs number of radars for collocated and asynchronous radar network, $L=1, K=40$}
	\label{fig:Async_Col_L1_OPT_Theta_N}
\end{figure}
\par In Figure~\ref{fig:Async_Col_L1_OPT_Theta_KKK} we demonstrate how $p_{opt}$ and $\Theta_{opt}$ vary with $K$. Towards this we fix $M=40$. From Figure~\ref{fig:Async_Col_L1_Theta_vs_p} we observed $p_{opt}=1$ for CSMA. From Figure~\ref{fig:Async_Col_L1_OPT_Theta_KKK}b we observe that $\Theta_{opt}$ increases in $K$ though the increase is small. From Figure~\ref{fig:Async_Col_L1_OPT_Theta_KKK}b we observe that in ALOHA $p_{opt}$ increases in $K$. This is because for a fixed $p$ the probability of interference decreases with the increase in chirp length, $K$. So, we can increase $p$ to achieve optimum throughput.
\begin{figure}	
	\centering
	\includegraphics[width=8cm, height=3.9cm]{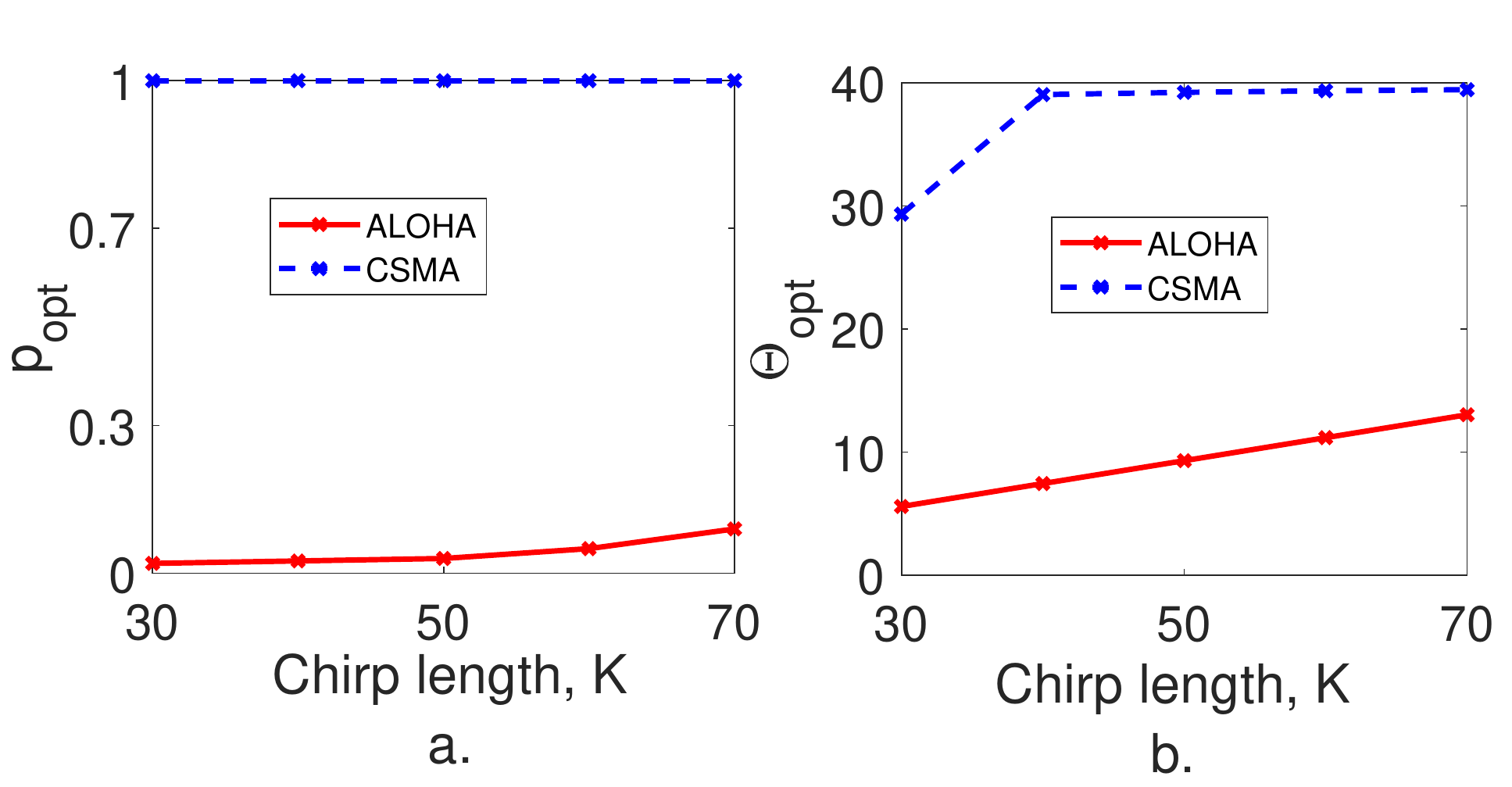}
	\caption{a. $ p_{opt}$ vs chirp length, b. $\Theta_{opt}$ vs chirp length for collocated and asynchronous radar network, $L=1, M=40$}
	\label{fig:Async_Col_L1_OPT_Theta_KKK}
\end{figure}
\subsubsection{$L>1$}
To begin with, we demonstrate how the throughput varies with $ p $. Towards this, we fix $K=40,L=4$ and vary $M$ as shown in Figure~\ref{fig:Async_Col_L4_Theta_vs_p}. In both the protocols we observe that throughput is less compared with $L=1$~(Figure~\ref{fig:Async_Col_L1_Theta_vs_p}). We explain this in the following. In CSMA for $L=1$ the radars suffer no interference (see Section~\ref{subsubsec:collocated-L=1}). For $L>1$ radars suffer interference and this deteriorates the throughput. In ALOHA for a given $p,M$, interference probability increases in $L$. So, the throughput is less compared to the case $L=1$. 
\begin{figure}	
	\centering
	\includegraphics[width=5.5cm, height=4cm]{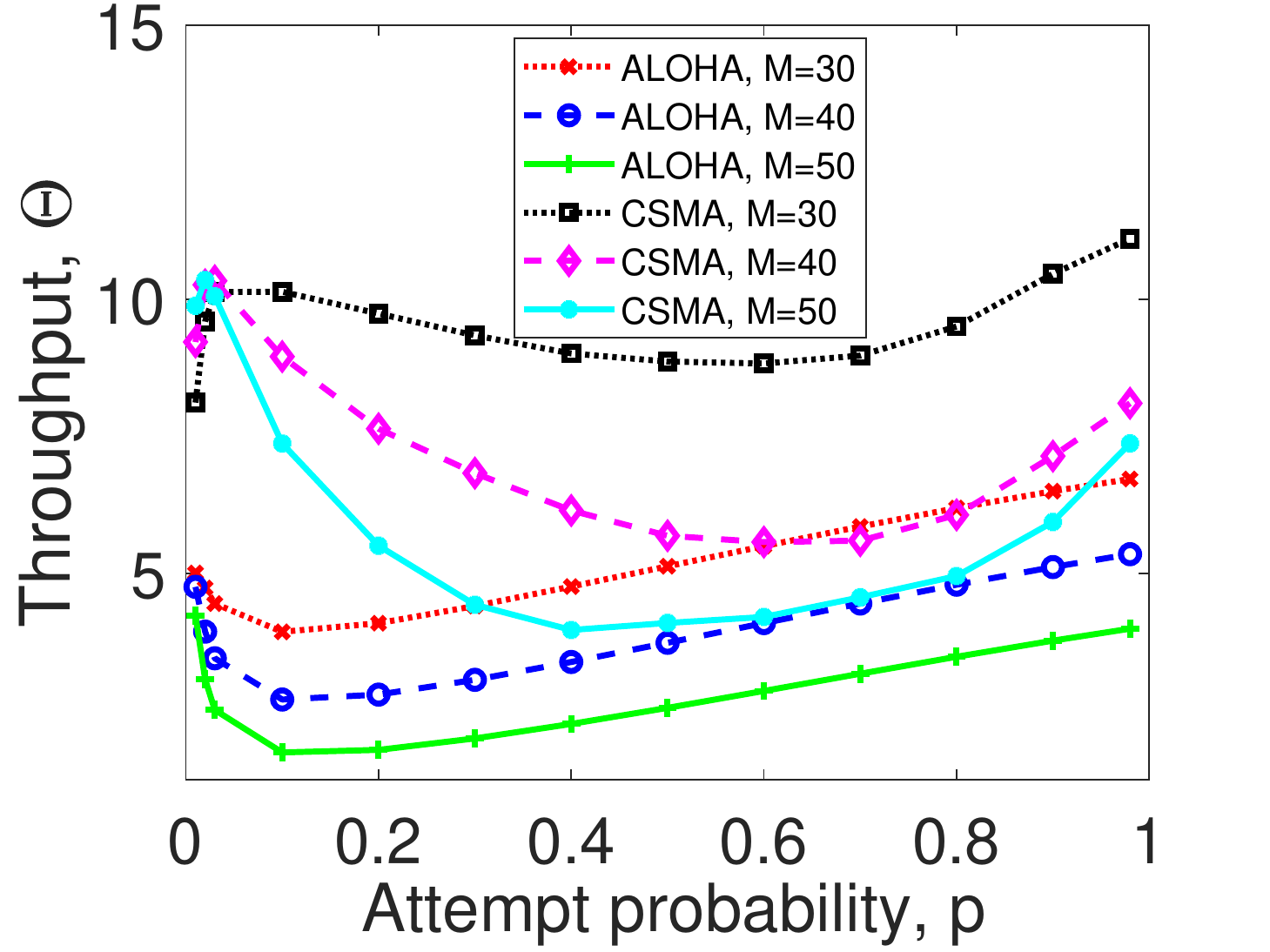}
	\caption{Throughput vs attempt probability for collocated and asynchronous radar network, $L=4, K=40, M = 30, 40, 50$
	}
	\label{fig:Async_Col_L4_Theta_vs_p}
\end{figure}
\par In Figure~\ref{fig:Async_Col_L4_OPT_Theta_N} we demonstrate how $p_{opt}$ and $\Theta_{opt}$ vary with $M$ in the network for CSMA and ALOHA. In Figure~\ref{fig:Async_Col_L4_OPT_Theta_N}a we observe that $p_{opt}$ decreases with the increase in $M$ in both CSMA and ALOHA. This is because, for a given $p$, the interference probability increases with the increase in $M$. So, the radar has to participate less aggressively in transmitting the packets. Further, in Figure~\ref{fig:Async_Col_L4_OPT_Theta_N}b we observe a decrease in throughput with the increase in $M$ for both CSMA and ALOHA. This is because for a fixed $p$ interference probability increases in $M$ so the throughput decreases. 
\begin{figure}	
	\centering
	\includegraphics[width=8cm, height=3.8cm]{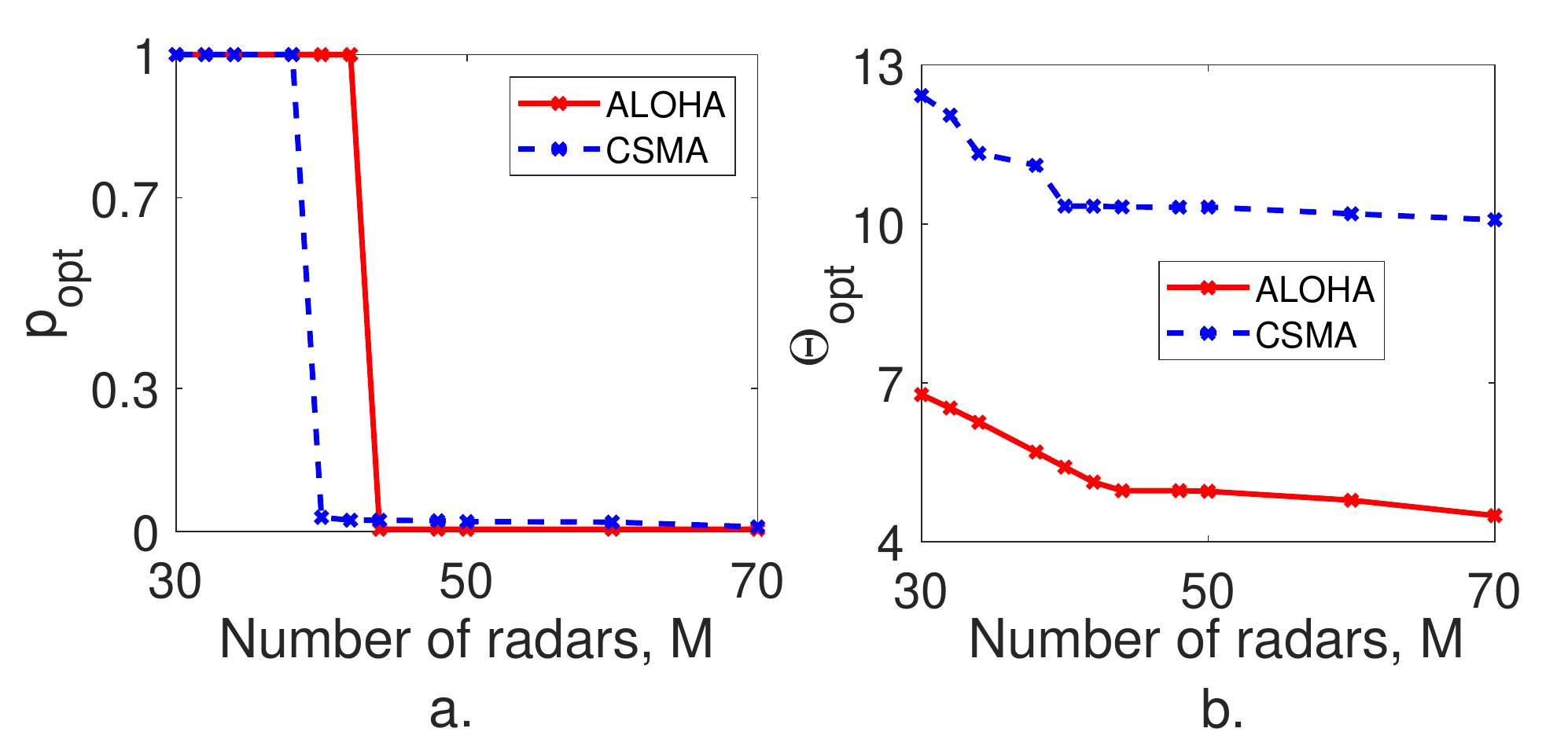}
	\caption{a. $ p_{opt}$ vs number of radars, b. $\Theta_{opt}$ vs number of radars for collocated and asynchronous radar network, $L=4, K=40$	
	}
	\label{fig:Async_Col_L4_OPT_Theta_N}
\end{figure}
\par In Figure~\ref{fig:Async_Col_L4_OPT_Theta_KKK} we demonstrate how $p_{opt}$ and $\Theta_{opt}$ vary with chirp length for CSMA and ALOHA. In Figure~\ref{fig:Async_Col_L4_OPT_Theta_KKK}a we observe that $p_{opt}$ increases in $K$ for both CSMA and ALOHA. This is because for a given $p,M$ interference probability decrease with chirp length, $K$. So, we can increase $p$ to achieve optimum throughput. In Figure~\ref{fig:Async_Col_L4_OPT_Theta_KKK}b we observe that $\Theta_{opt}$ increases in $K$ for both CSMA and ALOHA. 

\begin{figure}	
	\centering
	\includegraphics[width=8.75cm, height=3.8cm]{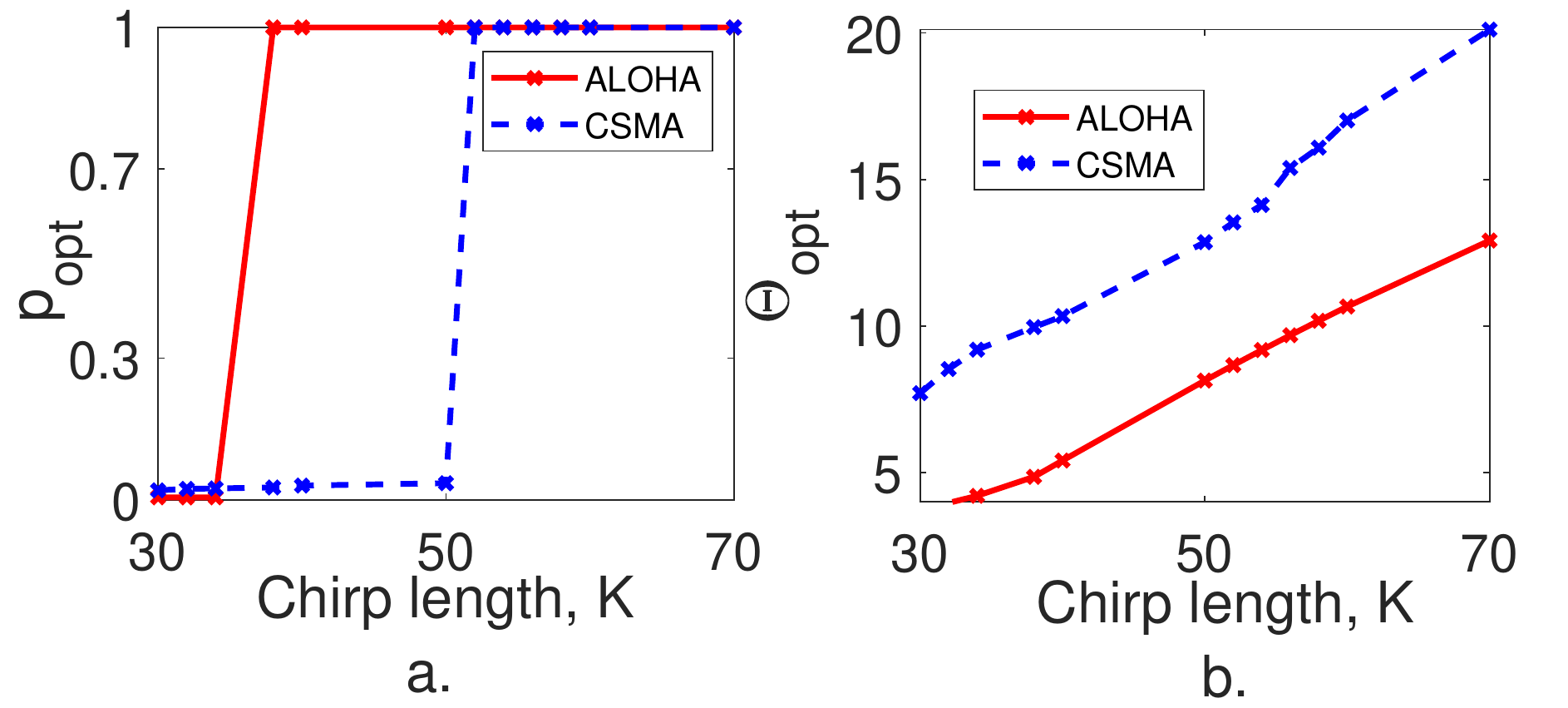}
	\caption{a. $ p_{opt}$ vs chirp length, b. $\Theta_{opt}$ vs chirp length for collocated and asynchronous radar network, $L=4, M=40,$}
	\label{fig:Async_Col_L4_OPT_Theta_KKK}
\end{figure}
\subsection{Non-collocated and asynchronous radar netwrok}
In this section, we simulate a collocated radar network. We first consider the case where the number of chirps per packet $ L = 1 $. Subsequently, we consider $ L > 1 $ case.
\subsubsection{$ L=1 $}
To begin with we demonstrate how the throughput varies with $ p $. Towards this we fix $K=40$ and vary $M$ as shown in Figure~\ref{fig:Non_Col_L1_Theta_vs_p}. We observe that in CSMA the throughput increases in $p$. However, it is less compared to the case $L=1$ collocated radar network at all the values of $p$~(see Figure~\ref{fig:Async_Col_L1_Theta_vs_p}). This is because in CSMA, the radars do not suffer interference when they are collocated but they suffer interference when they are not collocated. In ALOHA the throughputs are the same in collocated and non-collocated radar networks. This is because the radars do not sense the medium in ALOHA and each radar transmits packets irrespective of other radars' transmissions.
\begin{figure}	
	\centering
	\includegraphics[width=5.5cm, height=4cm]{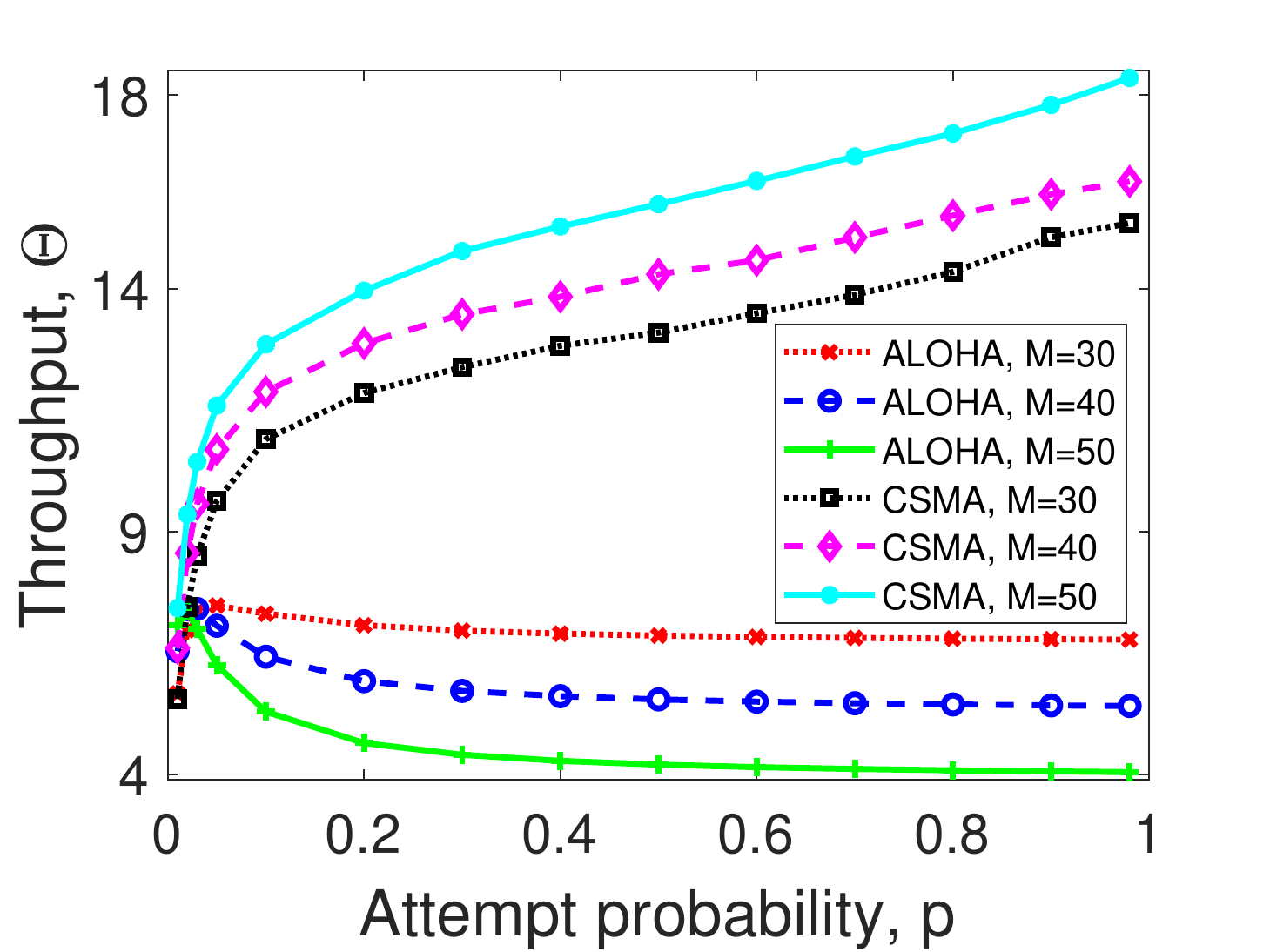}
	\caption{Throughput vs attempt probability for non-collocated and asynchronous radar network, $L=1, K=40, M = 30, 40, 50$
	}
	\label{fig:Non_Col_L1_Theta_vs_p}
\end{figure}

\par In Figure~\ref{fig:Non_Col_L1_OptTheta_vs_N} we demonstrate how $p_{opt}$ and $\Theta_{opt}$ vary with $M$ in CSMA and ALOHA. In Figure~\ref{fig:Non_Col_L1_Theta_vs_p} we observed that in CSMA throughput increases in $p$. So, in Figure~\ref{fig:Non_Col_L1_OptTheta_vs_N}a we observe $p_{opt}=1$ for CSMA. In Figure~\ref{fig:Non_Col_L1_OptTheta_vs_N}b we observe that throughput increases in $ M $, in CSMA. In ALOHA, given network parameters $L,K,M,p$ throughputs are the same for collocated and non-collocated networks. So, the same observations made for collocated network in ALOHA apply here~(see Section~\ref{subsub:coll-L=1}).
\begin{figure}	
	\centering
	\includegraphics[width=8cm, height=3.8cm]{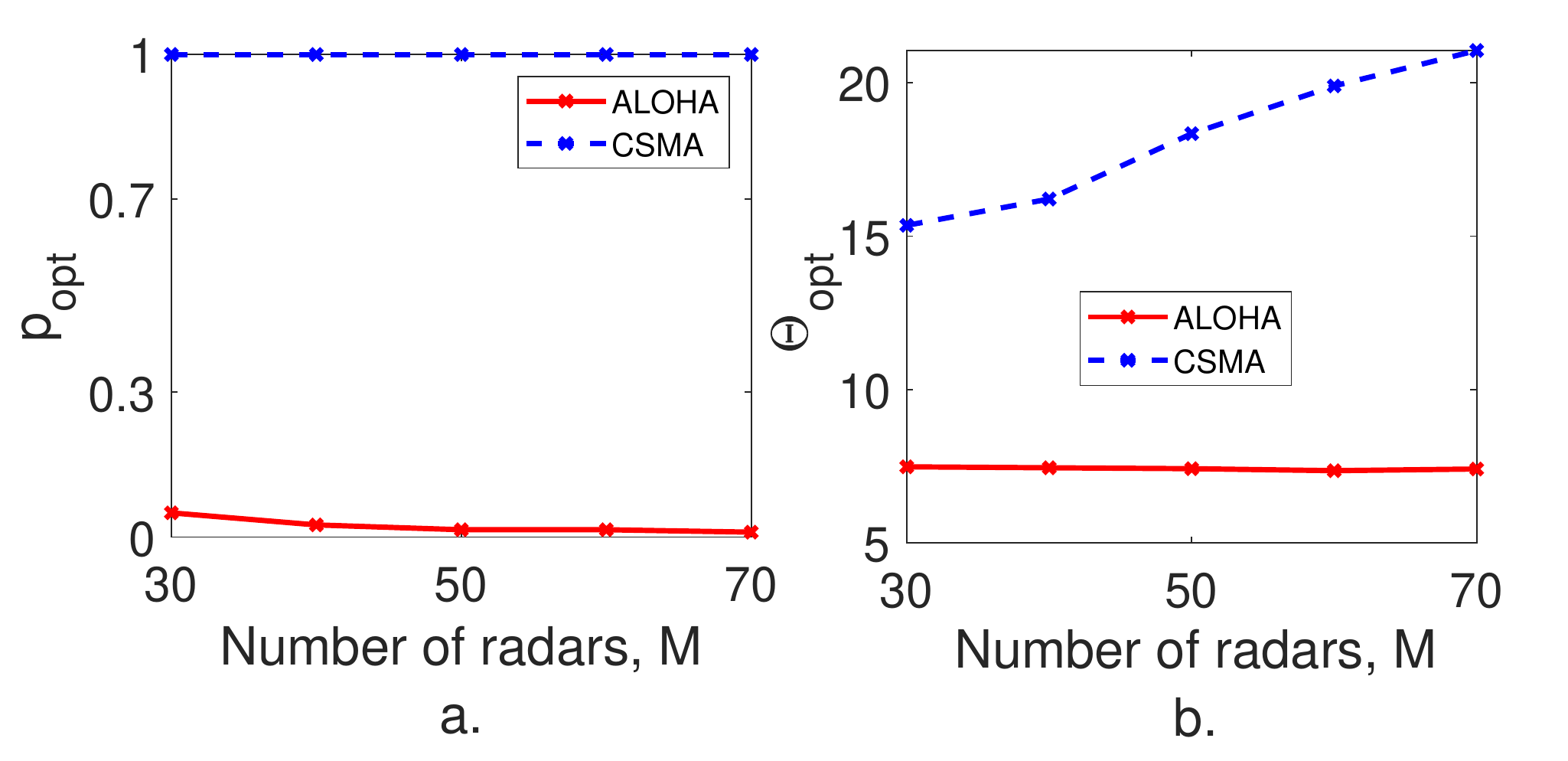}
	\caption{a. $ p_{opt}$ vs number of radars, b. $\Theta_{opt}$ vs number of radars for non-collocated and asynchronous radar network, $L=1, K=40$}
	\label{fig:Non_Col_L1_OptTheta_vs_N}
\end{figure}
\par In Figure~\ref{fig:Non_Col_L1_OptTheta_vs_K} we demonstrate how $p_{opt}$ and $\Theta_{opt}$ vary with chirp length in CSMA and ALOHA. In Figure~\ref{fig:Non_Col_L1_OptTheta_vs_K}a we observe that $p_{opt}=1$ for CSMA. In Figure~\ref{fig:Non_Col_L1_OptTheta_vs_K}b we observe that throughput increases in $K$ in CSMA. This is because the probability of interference decreases with the increase in $K$. So, the throughput increases.
\begin{figure}	
	\centering
	\includegraphics[width=8cm, height=3.8cm]{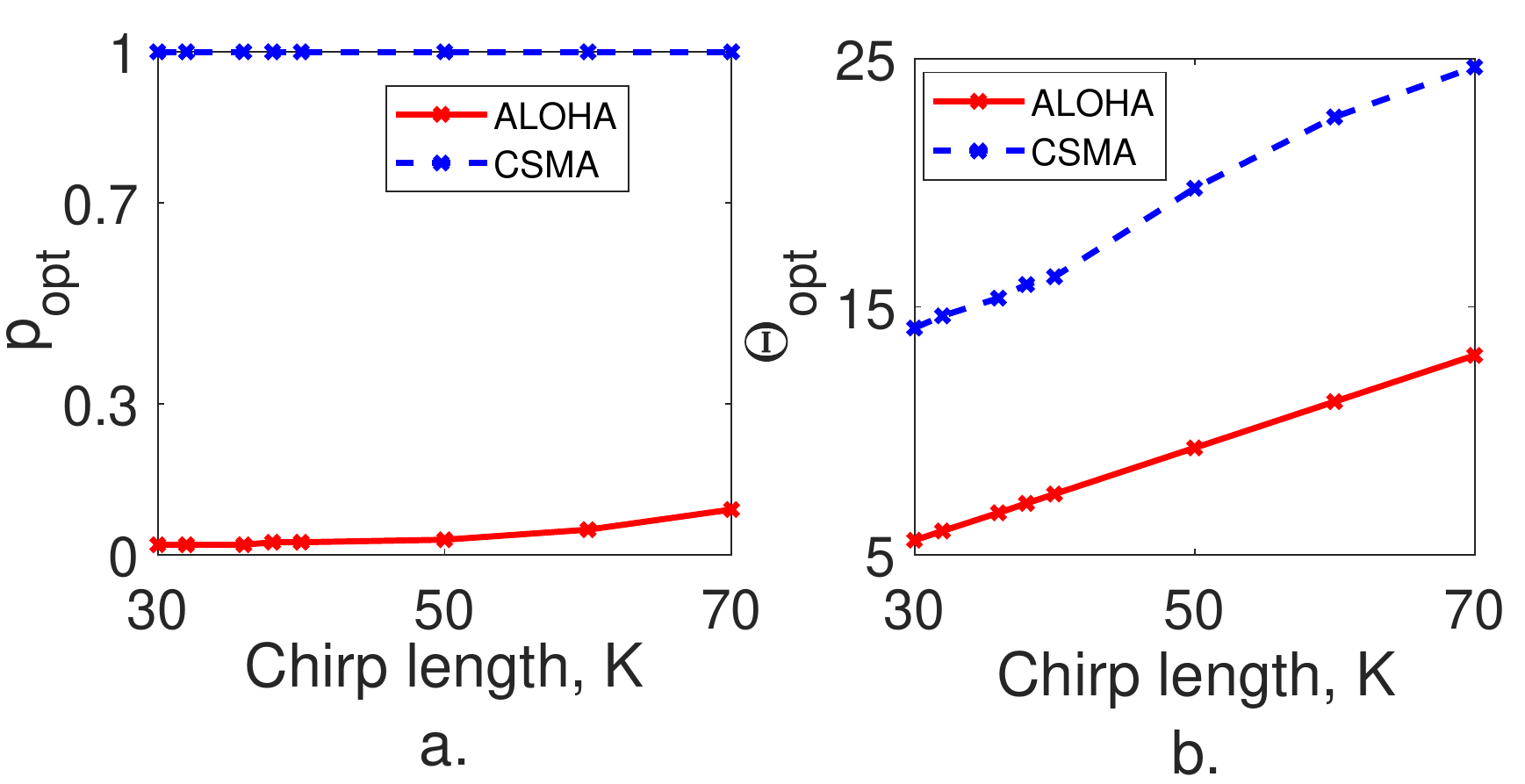}
	\caption{a. $ p_{opt}$ vs chirp length, b. $\Theta_{opt}$ vs chirp length for non-collocated and asynchronous radar network, $L=1, M=40,$}
	\label{fig:Non_Col_L1_OptTheta_vs_K}
\end{figure}
\subsubsection{$ L>1 $}
To begin with, we demonstrate how the throughput varies with $ p $. Towards this, we fix $K=40,L=4$ and vary $M$ as shown in Figure~\ref{fig:Non_Col_L4_Theta_vs_p}. In Figure~\ref{fig:Non_Col_L4_Theta_vs_p} we observe that the throughput in CSMA is less compared to the case of collocated radar network $L=4$~at all the values of $p$(see Figure~\ref{fig:Async_Col_L4_Theta_vs_p}). This is because the probability of interference increases in $L$. So, the throughput decreases. In ALOHA, given network parameters, $L,K,M,p$ throughputs are the same for collocated and non-collocated networks. So, the same observations made for collocated networks in ALOHA are hold here~(see Section~\ref{subsubsec:coll-L>1}.
\begin{figure}	
	\centering
	\includegraphics[width=5.5cm, height=4cm]{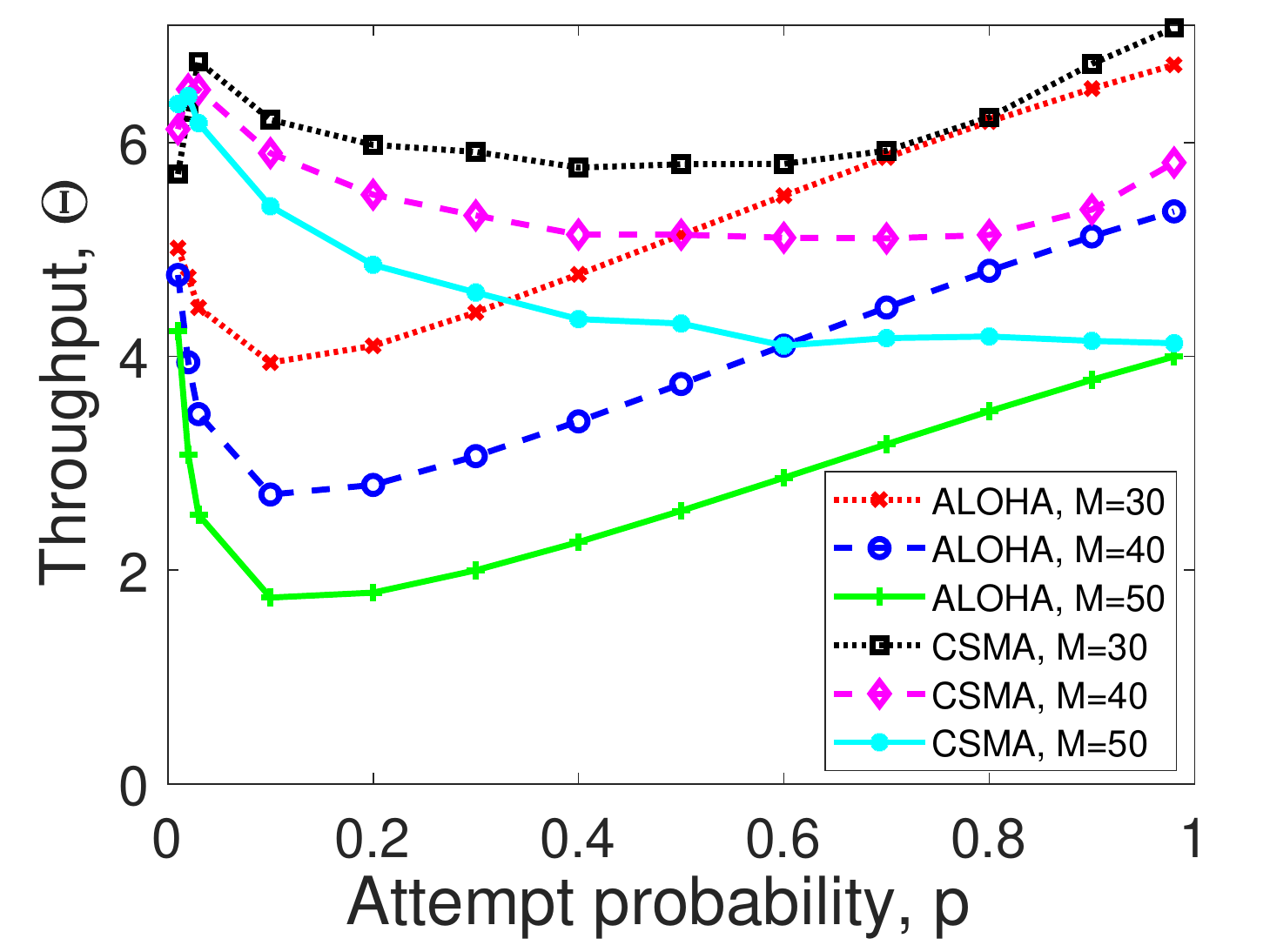}
	\caption{Throughput of slotted CSMA and slotted ALOHA for non-collocated radar network with asynchronous clock, $L=4, K=40, M = 30, 40, 50$}
	\label{fig:Non_Col_L4_Theta_vs_p}
\end{figure}
\par In Figure~\ref{fig:Non_Col_L4_OptTheta_vs_N} we demonstrate how $p_{opt}$ and $\Theta_{opt}$ vary with $M$ in CSMA and ALOHA. In Figure~\ref{fig:Non_Col_L4_OptTheta_vs_N}a we observe that $p_{opt}$ decreases with $M$ in both CSMA and ALOHA. This is because as the number of radars increases in the network the interference probability also increases. So, the radar has to participate less aggressively in transmitting the packets. Further, in Figure~\ref{fig:Async_Col_L4_OPT_Theta_N}b we observe a decrease in throughput with the increase in the number of radars, $M$ for both CSMA and ALOHA. This is because for a given $p$ the interference probability increases in $M$. So, the throughput decreases. 
\begin{figure}	
	\centering
	\includegraphics[width=7.5cm, height=4cm]{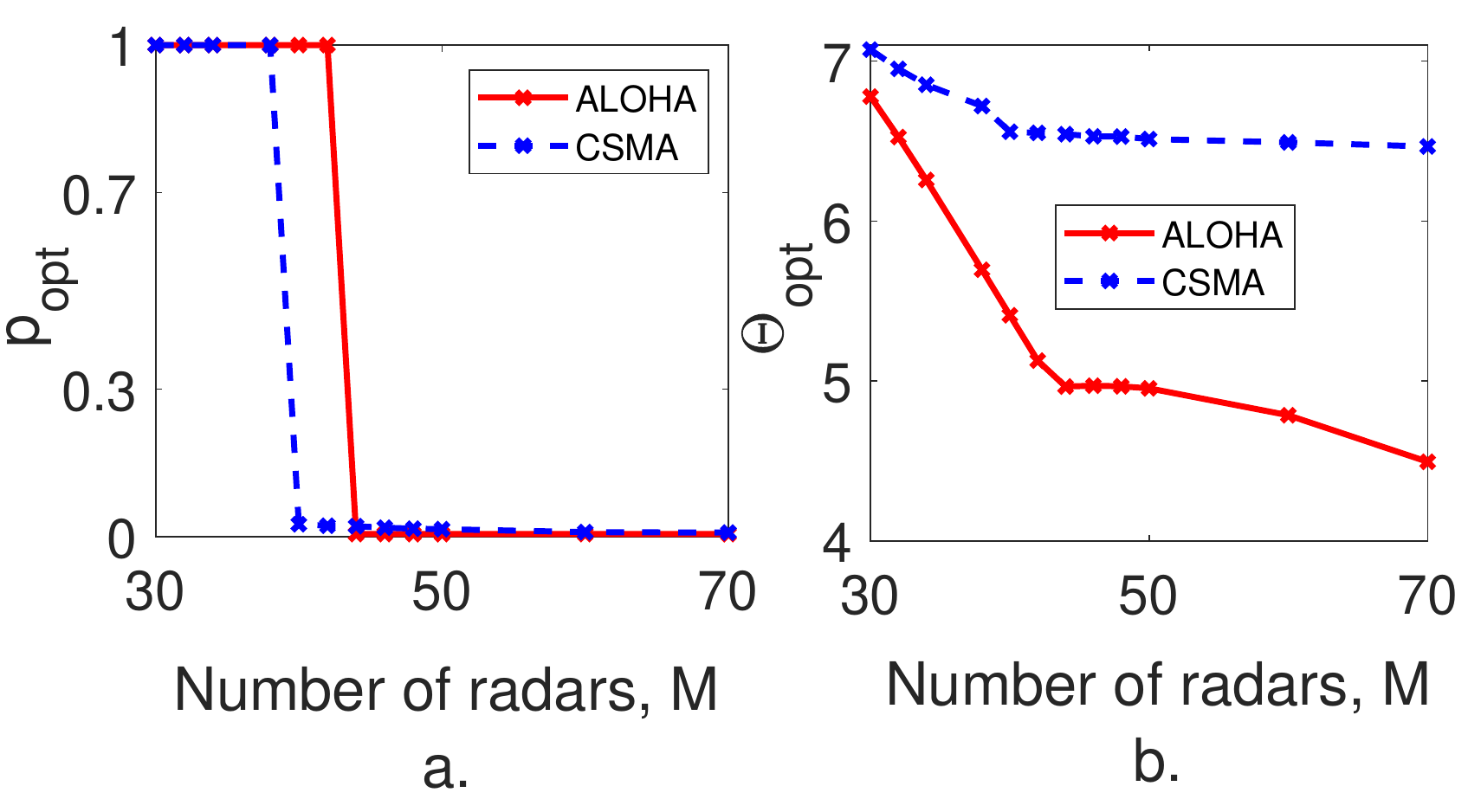}
	\caption{a. $ p_{opt}$ vs number of radars, b. $\Theta_{opt}$ vs number of radars for non-collocated and asynchronous radar network, $L=4, K=40$}
	\label{fig:Non_Col_L4_OptTheta_vs_N}
	\end{figure}
\par In Figure~\ref{fig:Non_Col_L4_OptTheta_vs_K} we demonstrate how $p_{opt}$ and $\Theta_{opt}$ vary with chirp length for CSMA and ALOHA. In Figure~\ref{fig:Non_Col_L4_OptTheta_vs_K}a we observe that $p_{opt}$ increases in $K$ for both CSMA and ALOHA. This is because, for a given $p,M$ interference probability decreases with the increase in $K$. So, radars can aggressively participate in the packet transmission and we can improve the throughput. So, in Figure~\ref{fig:Non_Col_L4_OptTheta_vs_K}b we observe the increase in $\Theta_{opt}$ with chirp length for both CSMA and ALOHA. 
\begin{figure}	
	\centering
	\includegraphics[width=8cm, height=3.7cm]{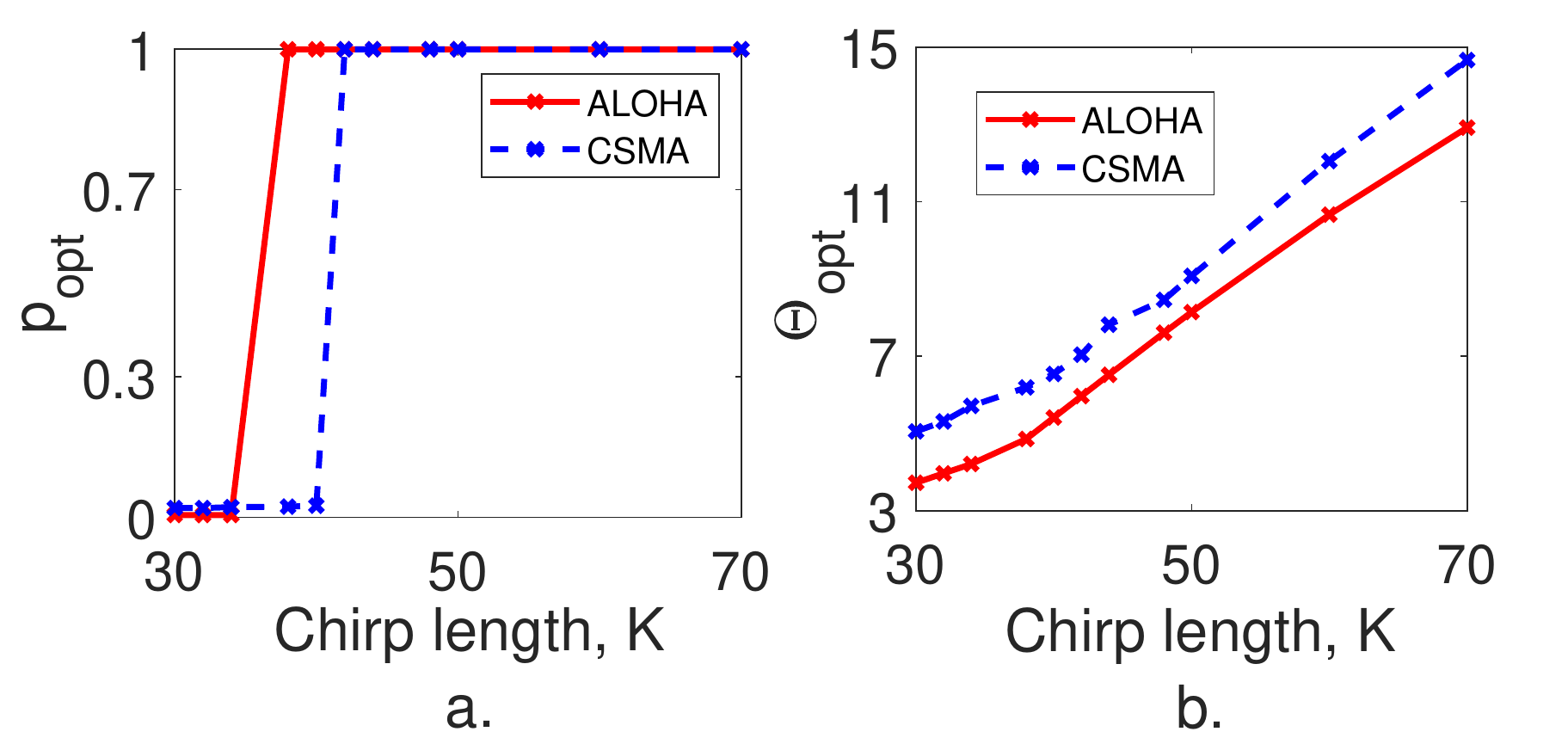}
	\caption{a. $ p_{opt}$ vs chirp length, b. $\Theta_{opt}$ vs chirp length for non-collocated and asynchronous radar network, $L=4, M=40$}
	\label{fig:Non_Col_L4_OptTheta_vs_K}
\end{figure}
\section{Conclusion}
We studied medium access in FMCW radar networks in which all the radars used the same parameters, e.g., chirp duration, chirp slope, cutoff frequency, number of chirps per packet, etc. We proposed and analyzed slotted ALOHA and CSMA protocols in terms of interference probability and throughput. In either case, we observed that interference probability and throughput may behave differently than in wireless communication networks. For 
instance, in the case of ALOHA, if the number of chirps per packet is larger than one, the interference probabilities may be smaller for higher transmission rates (see Figure~\ref{fig:ALOHA_1}a) and throughputs may be maximum at the highest possible transmission rates (see Figure~\ref{fig:ALOHA_1}b). In the case of CSMA also, using the highest possible attempt rates may maximize throughput (see Figure~\ref{fig:Async_Col_L1_Theta_vs_p}). We observed that CSMA outperformed ALOHA in all the realistic scenarios. 

If the radars in a FMCW network use different chirp slopes, they may also be subject to wideband interference from other radars. Designing medium access protocols that also consider wideband interference is a challenging open problem. Usually, the received signals at the radars undergo further processing for estimation of object locations, velocities, etc (see \cite{vtc_radar}). The notions of interference probability and throughput can be refined to account for the detection and estimation errors. Designing medium access protocols that optimize the refined notion of throughput is another potential future direction.

\bibliography{TSLS}
\end{document}